\documentclass[12pt]{article}
\usepackage{tikz}
\usetikzlibrary{shapes.geometric, arrows, positioning, calc}
\usepackage{hyperref}
\usepackage[utf8]{inputenc}
\usepackage{siunitx}
\usepackage{graphicx}
\usepackage{subcaption} 
\usepackage{booktabs}
\usepackage{multirow}
\usepackage{empheq}
\usepackage{amsmath}
\usepackage{enumerate}
\usepackage{amssymb}
\usepackage{geometry}
\usepackage{setspace}
\usepackage{graphicx}
\usepackage[authoryear,round]{natbib}
\usepackage{lmodern}
\usepackage{soul}
\usepackage{comment}
\usepackage{ifthen}
\usepackage{xr}
\usepackage{framed}
\usepackage{algorithm}
\usepackage{algpseudocode}
\newtheorem{theorem}{Theorem}[section]
\newtheorem{corollary}{Corollary}[section]

\newtheorem{lemma}{Lemma}[section]
\newtheorem{remark}{Remark}[section]
\newtheorem{proposition}{Proposition}[section]
\newtheorem{assumption}{Assumption}[section]
\newenvironment{proof}[1][Proof]{\textbf{#1.} }{\ \rule{0.5em}{0.5em}}
\DeclareMathSizes{12}{11}{7.7}{5.5}
\usepackage[inline]{enumitem}
\usepackage{comment}

\graphicspath{ {WD-Graphs/} }

\begin{document}

\title{Deep Penalty Methods: A Class of Deep Learning Algorithms for Solving High
Dimensional Optimal Stopping Problems}


\thispagestyle{empty}
\author{Yunfei Peng \and Pengyu Wei \and Wei Wei \and Xiaole Xue
\thanks{
Yunfei Peng:  School of Mathematics and Statistics, Guizhou University, Guiyang 550025, China. Email: yfpeng@gzu.edu.cn.
Pengyu Wei: Division of Banking \& Finance, Nanyang Business School, Nanyang Technological University,
Singapore. Email: pengyu.wei@ntu.edu.sg.
 Wei Wei: Department of Actuarial Mathematics and Statistics, School of Mathematics and Computer Science, Heriot-Watt University and  Maxwell Institute for Mathematical Sciences, Edinburgh, Scotland, EH14, 4AS, UK. E-mail: wei.wei@hw.ac.uk. 
 Xiaole Xue: School of Management, Shandong University, Jinan 250100, China.
 Email: xlxue@sdu.edu.cn.
 We thank seminar audiences at Nanyang Technological University, Heriot-Watt University, Warwick University and conference participants of Symposium on the Development of Partial Differential Equations in Shanghai, the Workshop on Stochastic Control, Financial Technology, and Machine Learning in Hong Kong, International Conference on Optimal Control and Machine Learning in Guizhou for helpful comments. 
}
}
%
%
\thispagestyle{empty}
\date{\today}
\maketitle
\thispagestyle{empty}

\vspace*{-0.4cm}
\begin{abstract}
We propose a deep learning algorithm for high dimensional optimal stopping problems. Our method is
inspired by the penalty method for solving variational inequalities. Within our approach, the penalized partial differential equation (PDE) is approximated using the deep backward stochastic differential equation (BSDE) framework proposed by \cite{weinan2017deep}, which leads us to coin the term ``Deep Penalty Method (DPM)" to refer to our algorithm. We show that the error of the DPM can be bounded by the cost function and $O(\frac{1}{\lambda})+O(\lambda h) +O(\sqrt{h})$, where $h$ is the step size in time and $\lambda$ is the penalty parameter. This finding emphasizes the need for careful consideration when selecting the penalization parameter and suggests that the discretization error converges at a rate of order $\frac{1}{2}$. We validate the efficacy of the DPM through numerical tests conducted on a high-dimensional optimal stopping model in the area of American option pricing. The numerical tests confirm both the accuracy and the computational efficiency of our proposed algorithm. 

\end{abstract}


\newpage
\setcounter{page}{1}
\section{Introduction}\label{SecIntroduction}
High dimensional optimal stopping models (e.g., American option pricing) have posed a long-standing computational challenge. In this paper, inspired by the penalty method for solving variational inequalities, we propose a deep learning algorithm for high dimensional optimal stopping problems in a continuous time setting.  Given its natural integration of deep learning and the penalty method, we call our algorithm ``Deep Penalty Method (DPM)''. \par 
Recently,  the deep backward stochastic differential equation (BSDE) method, advocated by \cite{weinan2017deep}, \cite{han2018solving} and \cite{beck2019machine}, has demonstrated a remarkable performance in addressing high dimensional dynamic models in continuous time. The Deep BSDE solver has been applied to tackle high dimensional optimal stopping problems, particularly American option pricing (e.g., \citeauthor{hure2020deep} \citeyear{hure2020deep}, \citeauthor{chen2021deep} \citeyear{chen2021deep} and \citeauthor{na2023efficient} \citeyear{na2023efficient}). In most of literature in this direction, an optimal stopping problem in continuous time is approximated by its discrete time companion with finitely many stopping opportunities. At each decision point, the Deep BSDE solver is utilized to obtain the continuation value, which is then compared with the stopping payoff. The maximum of the continuation value and the stopping payoff is then taken as the value function at the decision time and the terminal payoff in the next stage. This recursive process continues until time 0, yielding the value function of the optimal stopping problem.\par
It is easy to see that the aforementioned numerical methods for optimal stopping problems rely on discrete-time approximation, and the corresponding Deep BSDE solvers are implemented for each time step. On one hand, to control the ``discretization error", we need to increase the number of stopping opportunities. On the other hand, increasing the number of stopping opportunities will inevitably lead to the accumulation of error from the Deep BSDE solver at decision times\footnote{See the error analysis in \cite{hure2020deep}.}. As the optimization associated with the Deep BSDE solver is usually computationally costly, a delicate balance must be made between the ``discretization error" and the ``optimization error" inherent in the Deep BSDE solver.\par 
One of the advantages of the DPM is its ability to mitigate the accumulation of the ``optimization error''. The DPM integrates the Deep BSDE solver and penalty method, which essentially approximates the optimal stopping problems in continuous time by randomizing the stopping times using a sequence of Poisson arrival times\footnote{See, e.g., \cite{liang2015funding}, \cite{liang2015stochastic} and \cite{liang2013optimal} for rigorous mathematical discussions in various settings.}. As the penalty method yields a semi-linear PDE used to approximate the solution to the variational inequality arising from the optimal stopping problem, the optimization in Deep BSDE solver solely focuses on the terminal value of the penalized PDE and is executed only once. This effectively eliminates the accumulation of the errors due to the optimization in the Deep BSDE solver.\par 
In the error analysis, we show that the error of the DPM can be bounded by the cost function and $O(\frac{1}{\lambda})+O(\lambda h) +O(\sqrt{h})$, where $h$ represents the step size in time and $\lambda$ denotes the penalty parameter. In contrast to some commonly used numerical methods where $\lambda$ and $h$ can be chosen independently\footnote{See, e.g., \cite{forsyth2002quadratic} for more details.}, our findings suggest that the relationship between these parameters is critical when discretizing the penalized BSDE. Specifically, setting $\lambda = \frac{1}{\sqrt{h}}$ yields a convergence rate of $O(\sqrt{h})$ for the discretization. This result aligns with findings in the literature on numerical BSDEs using discrete-time approximations without penalty terms\footnote{See, e.g., \cite{bouchard2004discrete} and \cite{bouchard2008discrete} for the cases of BSDEs and reflected BSDEs respectively.}, and thus indicating that the penalty method utilized in the DPM would not exacerbate discretization errors. In addition, it is worth noting that the DPM can be extended to other models. In Appendix~B, we extend the method to optimal switching problems and provide a rigorous error analysis.\par 


Finally, we test our algorithm on a high-dimensional American index put option. This pricing problem for this type of option can be simplified to a standard one-dimensional option pricing case, and thus providing a good benchmark for our analysis. Through numerical experiments, we illustrate the effectiveness and accuracy of the DPM in pricing the American index option in high dimensions.\par 
\textit{Literature.} Numerous approaches have been developed to tackle optimal stopping models and the associated variational inequalities, such as binomial trees (e.g., \citeauthor{hul2003} \citeyear{hul2003}, \citeauthor{jiang2004convergence} \citeyear{jiang2004convergence} and \citeauthor{liang2007rate} \citeyear{liang2007rate}), penalty methods (e.g., \citeauthor{forsyth2002quadratic} \citeyear{forsyth2002quadratic} and \citeauthor{witte2011penalty} \citeyear{witte2011penalty}), policy iterations (e.g., \citeauthor{reisinger2012use} \citeyear{reisinger2012use}), least-squares Monte Carlo (LSM) methods (e.g., \citeauthor{longstaff2001valuing} \citeyear{longstaff2001valuing}) and stochastic mesh methods (e.g., \citeauthor{broadie2004stochastic} \citeyear{broadie2004stochastic}). While these numerical methods are widely employed in solving low ($\text{dimension}\le 3$) and /or moderate ($\text{dimension} \le 20$) dimensional optimal stopping problems, they become increasingly impractical for high-dimensional scenarios due to rapid growth in computational complexity.\par 
Recently, deep learning, which applies multi-layer neural networks to numerically solve optimization problems, has demonstrated impressive capabilities in solving high dimensional dynamic models.
Research in the literature has explored neural network-based algorithms for high-dimensional optimal stopping models. The usage of neural networks to numerically handle non-linear dynamic models traces back to \cite{lagaris1998artificial}, who approximate the solutions to PDEs using neural networks on a fixed mesh. \cite{haugh2004pricing} and \cite{kohler2010pricing} introduce neural networks with a single dense layer into American option pricing. The deep learning method for high dimensional PDEs is pioneered by \cite{weinan2017deep}, who establish the Deep BSDE solver. The error analysis of the Deep BSDE method for solving semi-linear PDEs is developed by \cite{han2020convergence}. Subsequently, the Deep BSDE solver has been refined in various ways (see, e.g., \citeauthor{raissi2018forward} \citeyear{raissi2018forward},  \citeauthor{chan2019machine} \citeyear{chan2019machine}, \citeauthor{hure2019some} \citeyear{hure2019some} and \citeauthor{sun2022credit} \citeyear{sun2022credit}).  Extensions of the Deep BSDE method to the realm of American option pricing has been undertaken by \cite{hure2020deep}, \cite{chen2021deep} and \cite{na2023efficient}.\par 
In addition to the Deep BSDE method, there are other avenues of deep learning based methods for high dimensional optimal stopping models. \cite{sirignano2018dgm} develop a systematic approach based on deep learning and the Galerkin method. \cite{dong2024randomized} and \cite{dai2024learning} establish a reinforcement learning framework for the study of optimal stopping problems. Different from the aforementioned literature which focuses on the value functions of the stopping problem, \cite{becker2019deep} employ deep learning approach to approximate optimal stopping rules.\par 
The remainder of the paper is organized as follows. Section \ref{GeneralModel} presents a general optimal stopping problem in a continuous time setting. The penalty method is demonstrated in Section \ref{sec:DPM}, and followed by an analysis of error resulting from the penalty approximation. Then the DPM is introduced and a detailed error analysis is provided. In Section \ref{sec:numerical}, we implement and test the algorithm on a high dimensional optimal stopping model arising in American option pricing. Section \ref{sec:conclusion} concludes. Proofs and extensions are collected in Appendix.

\section{General Setup }\label{GeneralModel}
\subsection{The dynamics}
We start with a filtered probability space $(\Omega,\mathcal{F},\{\mathcal{F}_t\}_{0\le t\le T},\mathbb{P})$ which satisfies the usual conditions. Let $W= (W_1,W_2,\ldots,W_n)^{\top}$ define an $n$-dimensional Brownian motion adapted to the filtration $\{\mathcal{F}_t\}_{t\ge0}.$  Filtration $\{\mathcal{F}_t\}_{t\ge0}$ captures all accessible information generated by a family of Markov process $X=X^x$ parameterized by the initial state $X_0=x\in\mathbb{R}^d$ and governed by the following
stochastic differential equation (SDE)\footnote{Any vector in this paper is a column vector without any further specification.},
\begin{equation}\label{SDE}
dX_t = b(t,X_t) dt + \sigma(t,X_t) dW_t
\end{equation}
where  $b(t,x) = (b_1(t,x),b_2(t,x),\ldots,b_d(t,x))^{\top}$, $\sigma(t,x) = (\sigma_{ij}(t,x))_{i=1,2,\ldots,d, j=1,2,\ldots,n}$ are functions defined on $[0,T]\times\mathbb{R}^d$ and valued in $\mathbb{R}^d$ and $\mathbb{R}^{d\times n}$ respectively. We introduce the following assumption for $b$ and $\sigma$. \par
\begin{assumption}\label{Assp:1}
 $b$ and $\sigma$ are uniformly H\"{o}lder continuous in $t$ with exponent $\frac{1}{2}$ and uniformly Lipschitz continuous in $x$, i.e., there exists an $L>0$ such that,
\begin{align}\label{Lip}
\vert \phi(t,x) - \phi(s,y)\vert\le L(\vert t-s \vert^{\frac{1}{2}}+\vert x - y\vert),
\end{align}
for all $t,s\in[0,T], x,y\in\mathbb{R}^d$, where $\phi\in\{b, \sigma\}, \vert A\vert: = \sqrt{\sum_{i=1}^d\sum_{j=1}^n a_{ij}^2}, \forall A:=(a_{ij})_{i=1,2,\cdots,d,j=i,2,\cdots n}\in \mathbb{R}^{d\times n}, n, d \in\mathbb{N}_+.$ 
\end{assumption}
This assumption guarantees the existence and uniqueness of the strong solution to SDE (\ref{SDE}). Moreover, it follows from standard arguments in SDE \citep[e.g.][chapter 1]{pham2009continuous} that, for any $T>0, p\ge 1$, there exists a constant $L$, such that
\begin{align}\label{Moment}
\mathbb{E}\left[\sup_{t\in[0,T]}\vert X_t \vert ^p \right]\le L(\vert x\vert^p+1),
\end{align}
where $L$ depends on $T$ and $p.$\par
For ease of notation, we introduce the differential operator associated with the process $X$ by
\begin{align}\label{diffOperator}
\mathcal{L}:=\frac{\partial}{\partial t} + \sum_{i=1}^d\sum_{j=1}^d a_{ij}(t,x)\frac{\partial^2}{\partial x_i\partial x_j}+\sum_{i=1}^d b_i(t,x)\frac{\partial}{\partial x_i},
\end{align}
where $a_{ij}(t,x) =\frac{1}{2} \sum_{k=1}^n\sigma_{ik}(t,x)\sigma_{jk}(t,x)$.\par 
Suppose that $\mathcal{L}$ is uniformly parabolic for any $T>0$, i.e., there exists a positive constant $\lambda_1>0$, such that for any $\xi=(\xi_1,\xi_2,\ldots,\xi_d)^{\top}\in\mathbb{R}^d,$
\begin{align}\label{Parabolic}
\sum_{i=1}^d\sum_{j=1}^da_{ij}(t,x)\xi_i\xi_j\ge \lambda_1\vert\xi\vert^2  ,\quad  \forall (t,x)\in[0,T]\times\mathbb{R}^d.
\end{align}

\subsection{The optimal stopping problem}
Consider an optimal stopping problem maturing at $T$. Denote the running payoff and the terminal payoff by $f: [0,T]\times\mathbb{R}^d\to\mathbb{R}$ and $g: \mathbb{R}^d\to\mathbb{R}$ respectively.
Define the stopping payoff by $p: [0,T]\times\mathbb{R}^d\to\mathbb{R},$ where $p\in C^{1,2}([0,T]\times\mathbb{R}^d)$. We make use of the following assumptions.
\begin{assumption}\label{Assp:2}
 $f$ is uniformly H\"{o}lder continuous in $t$ with exponent $\frac{1}{2}$. $f,g$ are uniformly Lipschitz continuous in $x$.
\end{assumption}
\begin{assumption}\label{Assp:3}
$p, \frac{\partial p}{\partial t}, \frac{\partial p}{\partial x_i}, \frac{\partial^2 p}{\partial x_i\partial x_j}, i,j=1,2,\ldots,m,$ are uniformly H\"{o}lder continuous in $t$ with exponent $\frac{1}{2}$ and uniformly Lipschitz continuous in $x$. Also, we suppose that
\begin{align*}
\inf_{(t,x)\in[0,T)\times\mathbb{R}^d}(\mathcal{L}p(t,x)-rp(t,x) +f(t,x))>-\infty.
\end{align*}
\end{assumption}
\begin{remark}
We will test our algorithm on an American index put option. It is easy to verify that the stopping and terminal payoff functions of the option satisfy these assumptions.
\end{remark}
 We consider the finite horizon optimal stopping time
 \begin{align}\label{ValueFunction}
 V(t,x) = \sup_{\tau\in\mathcal{T}_{t,T}}\mathbb{E}\left[\int^{\tau}_t e^{-r(s-t)}f(s,X_s)ds +1_{t\le \tau<T}e^{-r(\tau-t)}p(\tau,X_{\tau}) + 1_{\tau=T}e^{-r(T-t)}h(X_T) \vert X_t=x\right],
 \end{align}
 where $h(x) = \max\{g(x),p(T,x)\} $, $r$ is a positive constant and $\mathcal{T}_{t,T}$ denotes the set of stopping times valued in $[t,T].$\par 
 It follows from the standard argument in optimal stopping (e.g., \citeauthor{pham2009continuous} \citeyear{pham2009continuous} and \citeauthor{tou2013} \citeyear{tou2013}) that $V$ is the unique viscosity solution with quadratic growth to the following variational inequality
 \begin{eqnarray} \label{FBPDE}
\left\{\begin{array}{ll}
  \max\{\mathcal{L}V(t,x)-rV(t,x)+f(t,x), p(t,x)-V(t,x)\} = 0,\quad (t,x)\in[0,T)\times \mathbb{R}^d,\\
V(T,x) = h(x),\quad x\in\mathbb{R}^d.
\end{array}\right.
\end{eqnarray}
\section{The Algorithm and Error Analysis}\label{sec:DPM}
In this section, we first introduce the penalty representation for the variational inequality and then design a Deep BSDE algorithm for the resulting penalized PDE. Furthermore, our error analysis quantifies the dependence of the discretization error on the penalization parameter.
\subsection{The penalty approximation and penalization error}\label{sec:penalty}
The following penalty representation (\ref{PPDE}), reducing the variational inequality to a semi-linear PDEs, 
allows us to develop new algorithm based on the Deep BSDE method established by \cite{weinan2017deep}. 
We introduce the penalized PDE that approximates the variational inequality (\ref{FBPDE}) as follows,
\begin{eqnarray} \label{PPDE}
\left\{\begin{array}{ll}
 \mathcal{L}V^{\lambda}(t,x)-rV^{\lambda}(t,x)+f(t,x)+\lambda (p(t,x)-V^{\lambda}(t,x))^+ = 0\quad (t,x)\in[0,T)\times \mathbb{R}^d,\\
V^{\lambda}(T,x) = h(x),\quad x\in\mathbb{R}^d.
\end{array}\right.
\end{eqnarray}
\begin{proposition}\label{Prop:penaltyTerm}
Suppose that Assumptions \ref{Assp:1}-\ref{Assp:3} hold and $V^{\lambda}$ is the classical solution with quadratic growth to the PDE problem (\ref{PPDE}), then
\begin{align*}
\lambda (p(t,x)-V^{\lambda}(t,x))^+ \le C, \quad \forall (t,x,\lambda)\in[0,T]\times\mathbb{R}^d\times\mathbb{R}_+,
\end{align*}
where $C = -\inf_{(t,x)\in[0,T)\times\mathbb{R}^d}(\mathcal{L}p(t,x)-rp(t,x) +f(t,x)).$
\end{proposition}
\begin{theorem}\label{Prop:penaltyError}Suppose that $V^{\lambda}$ is the classical solution to the PDE problem (\ref{PPDE}) and $V$ defined by (\ref{ValueFunction}) is the value function of the optimal stopping time. Then
$$
0\le V(t,x)-V^{\lambda}(t,x)\le \frac{C}{\lambda},\forall (t,x,\lambda)\in[0,T]\times\mathbb{R}^d\times\mathbb{R}_+,
$$ where  $C = -\inf_{(t,x)\in[0,T)\times\mathbb{R}^d}(\mathcal{L}p(t,x)-rp(t,x) +f(t,x)).$
\end{theorem}
\begin{remark} Theorem \ref{Prop:penaltyError} shows the error bound of the penalty approximation. Results in this direction in various settings have appeared in literature. For example, \cite{liang2015stochastic} obtains that the error (in $L^2$ sense) due to the penalty approximation is bounded by $O(\frac{1}{\sqrt{\lambda}})$ in a general non-Markovian model.
\end{remark}

Before introducing the algorithm, let us make the following transformation. Let $U^{\lambda}(t,x): = (V^{\lambda}(t,x) - p(t,x))e^{-rt}, f_1(t,x) := (f(t,x) +\mathcal{L}p(t,x)-rp(t,x))e^{-rt}, h_1(x) := (h(x) - p(T,x))e^{-rT}, \forall (t,x)\in [0,T]\times\mathbb{R}^d.$ Then it follows from (\ref{PPDE}) that 
\begin{eqnarray} \label{MPPDE}
\left\{\begin{array}{ll}
 \mathcal{L}U^{\lambda}(t,x)+f_1(t,x)+\lambda (-U^{\lambda}(t,x))^+ = 0\quad (t,x)\in[0,T)\times \mathbb{R}^d,\\
U^{\lambda}(T,x) = h_1(x),\quad x\in\mathbb{R}^d.
\end{array}\right.
\end{eqnarray}
It is well known that the stochastic process pair $(Y^{\lambda}_t:=U^{\lambda}(t,X_t), Z^{\lambda}_t:=\sigma^T(t,X_t)\nabla_x U^{\lambda}(t,X_t))$, where $X_t$ is given by (\ref{SDE}), is the unique solution to the following BSDE problem\footnote{See, e.g., \cite{pardoux2005backward} and \cite{pardoux1999forward} for more details.},
\begin{eqnarray} \label{BSDE}
\left\{\begin{array}{ll}
 dY^{\lambda}_t = (-f_1(t,X_t) - \lambda (-Y^{\lambda}_t)^+)dt + (Z^{\lambda}_t)^{\top}dW_t,\\
Y^{\lambda}_T = h_1(X_T).
\end{array}\right.
\end{eqnarray}
We will focus on the numerical analysis of the BSDE problem (\ref{BSDE}) throughout this section.
\subsection{The discretization}\label{subsec:discretization} 
Let $\pi: 0=t_0<t_1<t_2<\dots<t_N = T$ be a partition of time interval $[0,T],$ with $h = \frac{T}{N}, t_i = ih, i = 0,1,2,\dots,N. $ The underlying dynamics $X$ is approximated by the classical Euler scheme.
\begin{align}\label{DDynamics}
X^{\pi}_{t_i} = X^{\pi}_{t_{i-1}} + b(t_{i-1},X^{\pi}_{t_{i-1}})h + \sigma(t_{i-1},X^{\pi}_{t_{i-1}})(W_{t_i}-W_{t_{i-1}}),\quad i = 1,2,\dots,N,
\end{align}
with $X^{\pi}_{t_0} =X_{t_0}.$\par 
Next we consider the implicit scheme for the BSDE (\ref{BSDE}).
\begin{align}\label{DDBSDE}
Y^{\lambda,\pi}_{t_{i+1}} = Y^{\lambda,\pi}_{t_{i}}- (f_1(t_i,X^{\pi}_{t_i}) + \lambda (-Y^{\lambda,\pi}_{t_i})^+)h + (Z^{\lambda,\pi }_{t_i})^{\top}(W_{t_{i+1}}-W_{t_i}),\quad \forall i = 0,1,2,\dots,N-1.
\end{align}
with $Y^{\lambda,\pi}_{t_N} = h_1(X_T^{\pi}), Z_{t_i}^{\lambda,\pi} = \frac{1}{h}\mathbb{E}_{t_i}[Y_{t_{i+1}}^{\lambda,\pi}(W_{t_{i+1}}-W_{t_i})].$\par 
\subsection{The network}
In contrast to the ``local'' approximation strategy introduced in the Deep BSDE framework by \cite{weinan2017deep} and \cite{han2018solving}---which employs a distinct neural network for each discrete time step $t_i$---we adopt a global approximation approach. By utilizing a single, integrated neural network $\mathcal{Z}$ to represent the function $Z^{\lambda,\pi}$ across the entire spatio-temporal domain, we define the approximation $Z^{\lambda,\pi}$ as a function of both time and state:

\begin{equation}\label{Network}
Z^{\lambda,\pi}(t_i, X^{\pi}_{t_i}) \approx \mathcal{Z}(t_i, X^{\pi}_{t_i} \mid \theta), \quad i = 0, 1, \dots, N-1,
\end{equation}
where $\theta = \{\theta_j\}_{j=1}^{J}$ denotes the vector of trainable parameters within the network. \par 
Equipping the discretized BSDE (\ref{DDBSDE}) with the neural network, we have the numerical scheme given by
\begin{align}\label{NetworkBSDE}
\mathcal{V}^{\lambda,\pi,\theta}_{t_{i+1}}  = \mathcal{V}^{\lambda,\pi,\theta}_{t_i} - f_1(t_i, X^{\pi}_{t_i})h- \lambda(-\mathcal{V}^{\lambda,\pi,\theta}_{t_i})^+ h+\mathcal{Z}(t_i,X^{\pi}_{t_i}\mid \theta)^{\top} ( W_{t_{i+1}} - W_{t_i})  , \quad \mathcal{V}^{\lambda,\pi,\theta}_0 = v.
\end{align}


The transition from a discrete local approximation to a global spatio-temporal network $\mathcal{Z}(\cdot, \cdot \mid \theta)$ offers a computational advantage over the traditional Deep BSDE framework. In conventional BSDE schemes, the GPU must frequently synchronize with the CPU to launch a new kernel for each time step’s specific network, leading to high latency. In contrast, the DPM utilizes a single spatio-temporal network $\mathcal{Z}(\cdot, \cdot \mid \theta)$. This not only mitigates the risk of memory exhaustion in high-dimensional settings but, more crucially, unlocks the power of spatio-temporal vectorization. By representing the solution as a single continuous function of both time and state, we can collapse the entire temporal dimension and batch dimension into a single composite input space during training. This enables the GPU to evaluate all time steps and batches across all simulated paths in a single, synchronized kernel execution. Unlike the local approach, which suffers from the cumulative latency of $N$ sequential CPU-GPU handshakes, our vectorized paradigm utilizes the hardware much more efficiently. This significantly enhances throughput and accelerates convergence by ensuring the GPU remains compute-bound rather than stalled by synchronization overhead. See Figure \ref{fig:flow_chart} for the data flow for the neutral network.\par 
\begin{figure}[htbp]
    \centering
    \includegraphics[width=\textwidth]{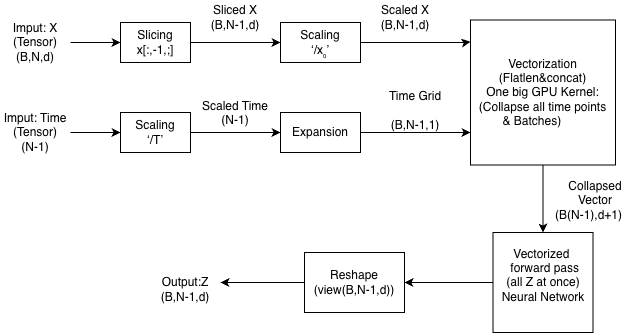}
    \caption{Data flow for the neural network for $\mathcal{Z}$. }
    \label{fig:flow_chart}
\end{figure}

\subsection{The cost function}
Define the cost function $\mathbb{E}\vert \mathcal{V}^{\lambda,\pi,\theta}_{t_N}-h_1(X^{\pi}_{t_N})\vert$ used in the Deep BSDE method  by $\iota(v,\theta)$, we then employ a stochastic optimizer to solve the following stochastic optimization problem
\begin{align}\label{Pr:ControNetwork}
\inf_{v\in\mathbb{R},\theta\in\mathbb{R}^J }\iota(v,\theta),
\end{align}
subject to (\ref{DDynamics}) and (\ref{NetworkBSDE}).\par
While conventional Deep BSDE solvers typically employ a Mean Squared Error (MSE) loss function, the DPM utilizes an $L^1$ loss in this paper. In the numerical experiments, we study the robustness of the choice of the loss functions by comparing the computational performance between the MSE and $L^1$ loss functions.

\subsection{The discretization error}
This section focuses on the analysis of the discretization error measured by $\sup\limits_{0\le i\le N}\mathbb{E}\vert Y^{\lambda,\pi}_{t_i} - Y^{\lambda}_{t_i}\vert$. We start with the following lemma which will be used to prove the main result. The proof of the lemma can be found in standard textbooks (see, e.g., \cite{platen2010numerical}), and thus we omit it.
\begin{lemma} \label{Lem:1}
Suppose that Assumptions \ref{Assp:1}-\ref{Assp:3} hold, then
\begin{align*}
&\mathbb{E} |\psi(t,X_t) - \psi(s,X_s)| \le K_1 \sqrt{s-t},\quad \forall\  0\le t<s \le T,\\
&\mathbb{E}\left[\sup_{t\in[0,T]}\vert\psi(t, X_t) \vert  \right]\le K_1,\quad \forall\  0\le t<s \le T,\\
&\mathbb{E} \left[\sup_{t\in[0,T]}|\psi(t,X_t) - \psi(t,X^{\pi}_t)|\right] \le K_1 \sqrt{h},
\end{align*}
where $\psi\in\{f_1,h_1\},$ $K_1$ is a constant depending on $T,X_0.$
\end{lemma}
\begin{theorem}\label{Thm:discretization}
 Suppose that Assumptions \ref{Assp:1}-\ref{Assp:3} hold, then
\begin{align*}
\sup_{0\le i\le N}\mathbb{E}\vert Y^{\lambda,\pi}_{t_i} - Y^{\lambda}_{t_i}\vert\le C_0 (\sqrt{h} + \lambda h), 
\end{align*}
where $C_0$ is a constant depending on $T, X_0.$
\end{theorem}\par 

While the selections of time step sizes and the penalization parameters are independent processes in some commonly used numerical methods for solving the variational inequality\footnote{For example, \cite{howison2013effect} employ a combination of the penalty method and finite difference method to numerically address a free boundary PDE problem. Their findings indicate that the discretization error is bounded by $(c+\frac{C}{\sqrt{\lambda}})(\Delta t + \Delta x^2),$ where $c,C$ are positive constants, $\Delta t$ is the time step size  and $\Delta x$ is the space step size. Consequently, the selections of $\lambda, \Delta t, \Delta x$ are independent, and hence one can choose arbitrarily large $\lambda$ without concern for $\Delta t, \Delta x.$}, 
Theorem \ref{Prop:penaltyError} and Theorem \ref{Thm:discretization} illustrate that the discretization error would be influenced by the selection of the ratios between the penalization parameter $\lambda$ and the time step size $h$. Hence the results in this paper remind one to give more careful consideration to the selection of the penalization parameter when implementing the DPM.
\par 
Moreover, Theorems \ref{Prop:penaltyError} \ref{Thm:discretization} suggest that the optimal error bound is $O(\sqrt{h})$ at most. Choosing $\lambda = \frac{1}{\sqrt{h}}$, we obtain the following result.\par 
\begin{corollary}\label{Cor:1}
 Suppose that Assumptions \ref{Assp:1}-\ref{Assp:3} hold, then
\begin{align*}
\sup_{0\le i\le N}\mathbb{E}\vert e^{rt_i}Y^{\lambda,\pi}_{t_i} +p(t_i, X_{t_i})- V(t_i,X_{t_i})\vert\le K\sqrt{h},
\end{align*}
where $K$ is a constant depending on $X_0,T$.
\end{corollary}

\subsection{Error bounds for the DPM}
In this section, we make analysis of the error bound for the DPM. We start with the following lemma, which shows that the difference of two discretization schemes at each step can be bounded by the difference at the terminal time.
\begin{lemma}\label{Lem:3}
Suppose that Assumptions  (\ref{Assp:1})-(\ref{Assp:3}) hold. Assume that $(\{Y^{\lambda,\pi,1}_{t_i}\}_{0\le i\le N},\{Z^{\lambda,\pi,1}_{t_i}\}_{0\le i\le N-1})$ and $(\{Y^{\lambda,\pi,2}_{t_i}\}_{0\le i\le N},\{Z^{\lambda,\pi,2}_{t_i}\}_{0\le i\le N-1})$ are discretization schemes that follow (\ref{DDBSDE}), then
\begin{align*}
\sup_{0\le i\le N}\mathbb{E}|\delta Y^{\lambda,\pi}_{t_i}|\le \mathbb{E}|\delta Y^{\lambda,\pi}_{t_N}|,
\end{align*}
where $\delta Y^{\lambda,\pi}_{t_i}:= Y^{\lambda,\pi,1}_{t_i}-Y^{\lambda,\pi,2}_{t_i}, 0\le i\le N$.
\end{lemma}

\begin{theorem}\label{Thm:errorBounds}
 Suppose that Assumptions \ref{Assp:1}-\ref{Assp:3} hold, then
 \begin{align*}
\sup_{0\le i\le N}\mathbb{E}\vert \mathcal{V}^{\lambda,\pi,\theta}_{t_{i}} - Y^{\lambda}_{t_i}\vert\le C_1 (\sqrt{h} + \lambda h)+  \mathbb{E}\vert \mathcal{V}^{\lambda,\pi,\theta}_{t_N}-h_1(X_{t_N})\vert,
 \end{align*}
 where $C_1$ is a constant depending on $T,X_0$.\par 
\end{theorem}
\begin{corollary}\label{corollary:3}
Suppose that Assumptions \ref{Assp:1}-\ref{Assp:3} hold, then
 \begin{align*}
\sup_{0\le i\le N}\mathbb{E}\vert \mathcal{V}^{\lambda,\pi,\theta}_{t_{i}}e^{rt_i}+p(t_i,X_{t_i}) - V({t_i},X_{t_i})\vert\le C_2 \sqrt{h} +  e^{rT}\mathbb{E}\vert \mathcal{V}^{\lambda,\pi,\theta}_{t_N}-h_1(X_{t_N})\vert,
 \end{align*}
 where $C_2$ is a constant depending on $T,X_0.$\par 
\end{corollary}
\section{Numerical Examples}\label{sec:numerical}
\subsection{The model}
In this section, we test our algorithm on an American index put option, where the index is given by the geometric average of the underlying assets. Specifically, we suppose $h(x) = (K-(\prod\limits_{i=1}^d x_i)^{\frac{1}{d}})^+, p(t,x) =  K-(\prod\limits_{i=1}^d x_i)^{\frac{1}{d}}, f(t,x) = 0,$ where $K$ is a positive constant and $d$ is the number of the underlying assets. Moreover, we suppose the dynamics of the underlying assets follows a multi-dimensional Geometric Brownian motion, 
\begin{align*}
dX_{it} = \mu X_{it}dt+\sigma X_{it}dW_{it}, i = 1,2,\dots,d,
\end{align*}
where $\mu, \sigma$ are constants and $W= (W_1,W_2,\ldots,W_d)^{\top}$  is a $d-$dimensional Brownian motion.\par 
It is easy to see that the pricing of the American option considered in this section can be simplified to the case of standard one-dimensional American put options. In more details, let $I_t = (\prod\limits_{i=1}^d X_{it})^{\frac{1}{d}}$, then the payoff functions $h(I_t) = (K-I_t)^+,\  p(t,I_t) =  K-I_t,\  f(t,I_t) = 0,$ and the dynamics of $I_t$ is given by
\begin{align*}
dI_t = \hat{\mu}I_tdt + \hat{\sigma}I_tdB_t,
\end{align*}
where $\hat{\mu} = \mu - \frac{1}{2}\sigma^2 + \frac{1}{2d}\sigma^2$, $\hat{\sigma} = \frac{\sigma}{\sqrt{d}}$ and $B=\{B_t\}_{t\ge0}$ is a one-dimensional Brownian motion.\par 
Thus, by the above transformation, the multi-dimensional optimal stopping problem is reduced to a standard one-dimensional American option pricing, where the underlying asset is given by $I_t.$ This observation allows us to use the finite difference method to obtain the benchmark solution to the optimal stopping problem. \par 
The code for the numerical tests in this section is available at: \\
\url{https://github.com/weiw1422/DeepPenaltyMethod/blob/main/DPM_convergence.ipynb}.
\subsection{Architectural specification and optimization of \texorpdfstring{$\mathcal{Z}(\cdot,\cdot\mid\theta)$}{Z(t,x|theta)}}
We employ the ResNet architecture, pioneered by \cite{he2016deep}, to parameterize the global spatio-temporal network $\mathcal{Z}(\cdot, \cdot \mid \theta)$ for the optimization problem in (\ref{Pr:ControNetwork}). Specifically, the network begins with an initial projection layer that maps the $(d+1)$-dimensional input—comprising both time and state variables—into a hidden feature space of dimension $h_f$. This is followed by $L$ residual blocks. Each block executes the transformation:
\begin{equation}
y = \text{LayerNorm}(x + 0.5 \cdot \mathcal{F}(x, \theta_l))
\end{equation}
where $\mathcal{F}$ denotes a sub-network consisting of two fully connected layers interleaved with the SiLU (Swish) activation function. By incorporating a scaling factor of $0.5$ and Layer Normalization, this architecture ensures robust gradient propagation and prevents numerical instability. The specific architectural configuration is detailed in Table \ref{tab:network_specs}.
\begin{table}[H]
\centering
\begin{tabular}{lcl}
\toprule
\textbf{Parameter} &  \textbf{Value} \\
\midrule
Hidden Units &  128  \\
Residual Blocks &  8 \\
Activation Function  & SiLU (Swish) \\
Normalization & Layer Normalization \\
Weight Initialization &  Xavier Uniform \\
\bottomrule
\end{tabular}
\caption{ResNet hyperparameters}
\label{tab:network_specs}
\end{table}
To ensure stable convergence of the global spatio-temporal network, we employ the Adam optimizer with an initial learning rate of $1 \times 10^{-3}$. To navigate the loss landscape, we implement a ``ReduceLROnPlateau'' scheduler. The scheduler monitors the validation loss and reduces the learning rate by a factor of $0.5$ if the loss remains stagnant for a patience period of $1,000$ iterations. The minimum learning rate is $1 \times 10^{-7}$.
\subsection{Convergence analysis}
Table \ref{TableBenchmark} presents a comparative analysis between the DPM approximations and the benchmarks obtained via the finite difference method. These results provide robust numerical evidence of the efficiency, accuracy, and scalability of the DPM across a wide range of dimensions, reaching up to $d=200$. The data reveals that the DPM maintains a relative error significantly below $1\%$ across all tested cases. Furthermore, the consistently low loss variance ($O(10^{-8})-O(10^{-7})$) across all high-dimensional tests underscores the stability of the DPM during the final stages of optimization. \par 
Figure \ref{fig:loss_convergence_grid} shows the training progression. The rapid decay and subsequent stabilization of the cost function across all dimensions demonstrate the robust learning capability of the DPM.
\begin{table}[H]
\centering
\begin{tabular}{l *{2}{S[table-format=1.4]} S[table-format=1.4e-1] S[table-format=1.4] S[table-format=1.4] S[table-format=2.2]}
\toprule
$d$ & {$\bar{V}$} & {$V_{b}$} & {Loss Variance} & {Rel. Error (\%)} & {Cost Value} \\ 
\midrule
10  & 1.4978 & 1.4958 & 3.1717e-8 & 0.1306 & 0.0087 \\
20   & 1.5188 & 1.5155 & 1.9505e-8 & 0.2192 & 0.0127\\
25  & 1.5228 & 1.5194 & 2.2140e-8 & 0.2251 & 0.0132  \\
50  & 1.5322 & 1.5270 & 3.4896e-8 & 0.3409 & 0.0178  \\
100 & 1.5342 & 1.5307 & 4.8503e-8 & 0.2313 & 0.0141  \\ 
200 & 1.5374 & 1.5326 & 1.0425e-7 & 0.3150 & 0.0172 \\
\bottomrule
\end{tabular}
\vspace{0.3cm}
\caption{The approximating solution to the variational inequality (\ref{FBPDE}).  \small\textit{$\bar{V}$ is the result of the last epoch used to approximate $V(0,1,1,\cdots,1)$. $V_b$ is the benchmark obtained with the finite difference method. The loss variance is the variance of a list of values of the cost functions of the last $1000$ epochs. The relative error is calculated by $\frac{|\bar{V}-V_{b}|}{V_b}$. The cost value is the value of the cost function for the last epoch.
Numerical experiments were performed on an NVIDIA G4 GPU with the following parameters. Problem parameters:  $r = \mu = 0.05$, $\sigma = \sqrt{2}$, $K = 2$, $T = 1$. Penalty and discretization: $N = 99$ time steps, penalty $\lambda = \frac{1}{\sqrt{T/N}}=9.95$. Training:  $30,000$ epochs, batch size $8,192$ (single-batch training per epoch).}}
\label{TableBenchmark}
\end{table}

\begin{figure}[htbp]
    \centering
    \includegraphics[width=\textwidth]{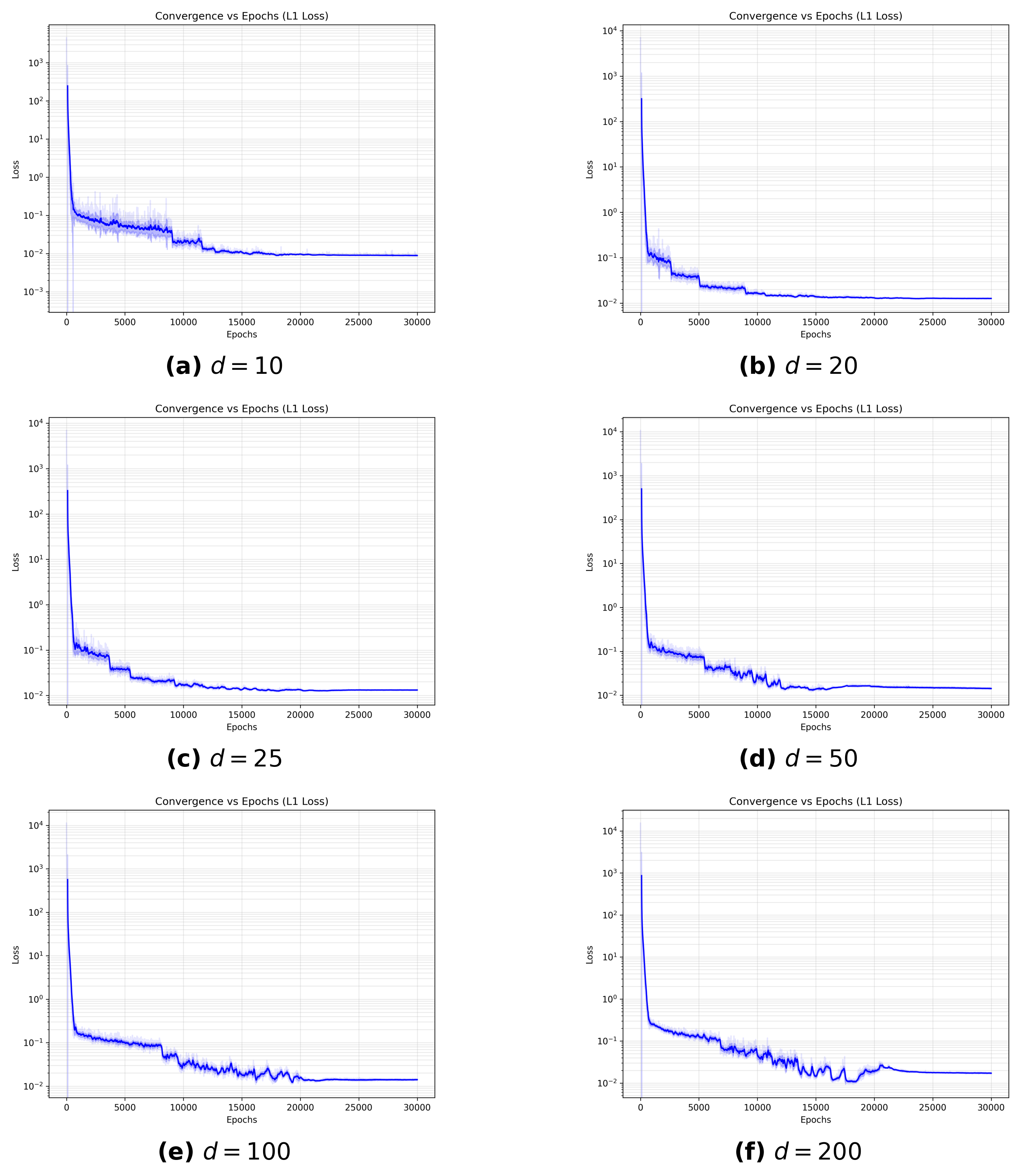}
    \caption{Training loss convergence across dimensions $d \in \{10, 20, 25, 50, 100, 200\}$. \small\textit{Each panel depicts the $L^1$ loss trajectory over 30,000 epochs. Numerical experiments were performed on an NVIDIA G4 GPU with the following parameters. Problem parameters:  $r = \mu = 0.05$, $\sigma = \sqrt{2}$, $K = 2$, $T = 1$. Penalty and discretization: $N = 99$ time steps, penalty $\lambda = \frac{1}{\sqrt{T/N}}=9.95$. Training:  $30,000$ epochs, batch size $8,192$ (single-batch training per epoch).}}
    \label{fig:loss_convergence_grid}
\end{figure}

\subsection{Computational efficiency}
The results reported in Table~\ref{tab:convergence_comparison} provide numerical evidence for the computational efficiency and scalability of the DPM in high-dimensional settings.\par 
The total training time exhibits only mild dependence on the dimension $d$, increasing from 21.29 minutes at $d = 10$ to 29.58 minutes at $d = 200$. The performance of the training time to increasing dimensionality is a consequence of non-recursive structure of the DPM which allows us to collapse the temporal dimension and batch dimension into a single composite input space. This weak scaling suggests that the method effectively exploits parallel hardware and that the overall computational cost remains consistent with quadratic or linear complexity in $d$\footnote{For the Deep BSDE solver with a ResNet architecture, the computational complexity per forward pass is $O(d \cdot h_f + h_f^2 \cdot L)$. Under this formulation, the complexity scales linearly with respect to the asset dimension $d$, provided the hidden width remains constant. However, in literature where the network width is scaled proportionally to the input dimension (i.e., $h_f \sim d$), the complexity becomes quadratic, $O(d^2 L)$. }.\par 
Beyond raw execution time, the \textit{Stable Entry time} ($t^*$) provides a rigorous metric for the time-to-solution. This represents the point after which the approximation remains permanently within the 1\% relative error band. The data reveals that $t^*$ scales very well, doubling from 8.02 to 16.32 minutes while the dimensionality increases twenty-fold. The consistently high efficiency ratios indicate that the solver stabilizes the solution well before the completion of the 30,000-epoch budget. For instance, in the $d=25$ case, stable convergence is attained in only 26.77\% of the total runtime.\par 

\begin{table}[H]
\centering
\small
\begin{tabular}{l c c c}
\toprule
\textbf{Dimension ($d$)} & \textbf{Total Time (min)} & \textbf{Stable Entry $t^*$ (min)} & \textbf{Efficiency Ratio (\%)} \\ 
\midrule
10  & 21.29 & 8.02  & 37.67 \\
20  & 25.05 & 9.19  & 36.69 \\
25  & 21.63 & 5.79  & 26.77 \\
50  & 23.47 & 9.94  & 42.35 \\
100 & 26.45 & 12.99 & 49.11 \\ 
200 & 29.58 & 16.32 & 55.17 \\
\bottomrule
\end{tabular}
\vspace{0.3cm}
\caption{Computational efficiency. \small\textit{The efficiency ratio ($t^* / \text{Total Time}$) quantifies the training overhead. The sub-linear growth of $t^*$ relative to $d$ highlights the efficiency of the vectorized solver on parallel hardware. Numerical experiments were performed on an NVIDIA G4 GPU with the following parameters. Problem parameters:  $r = \mu = 0.05$, $\sigma = \sqrt{2}$, $K = 2$, $T = 1$. Penalty and discretization: $N = 99$ time steps, penalty $\lambda = \frac{1}{\sqrt{T/N}}=9.95$. Training:  $30,000$ epochs, batch size $8,192$ (single-batch training per epoch).}}
\label{tab:convergence_comparison}
\end{table}
Figure \ref{fig:global_convergence_grid} shows the convergence of the solution to the variational inequality. Despite the high-dimensional state space, the solver achieves stable convergence in all cases, with total training time scaling sub-linearly relative to the dimension.
\begin{figure}[htbp]
    \centering
    \includegraphics[width=\textwidth]{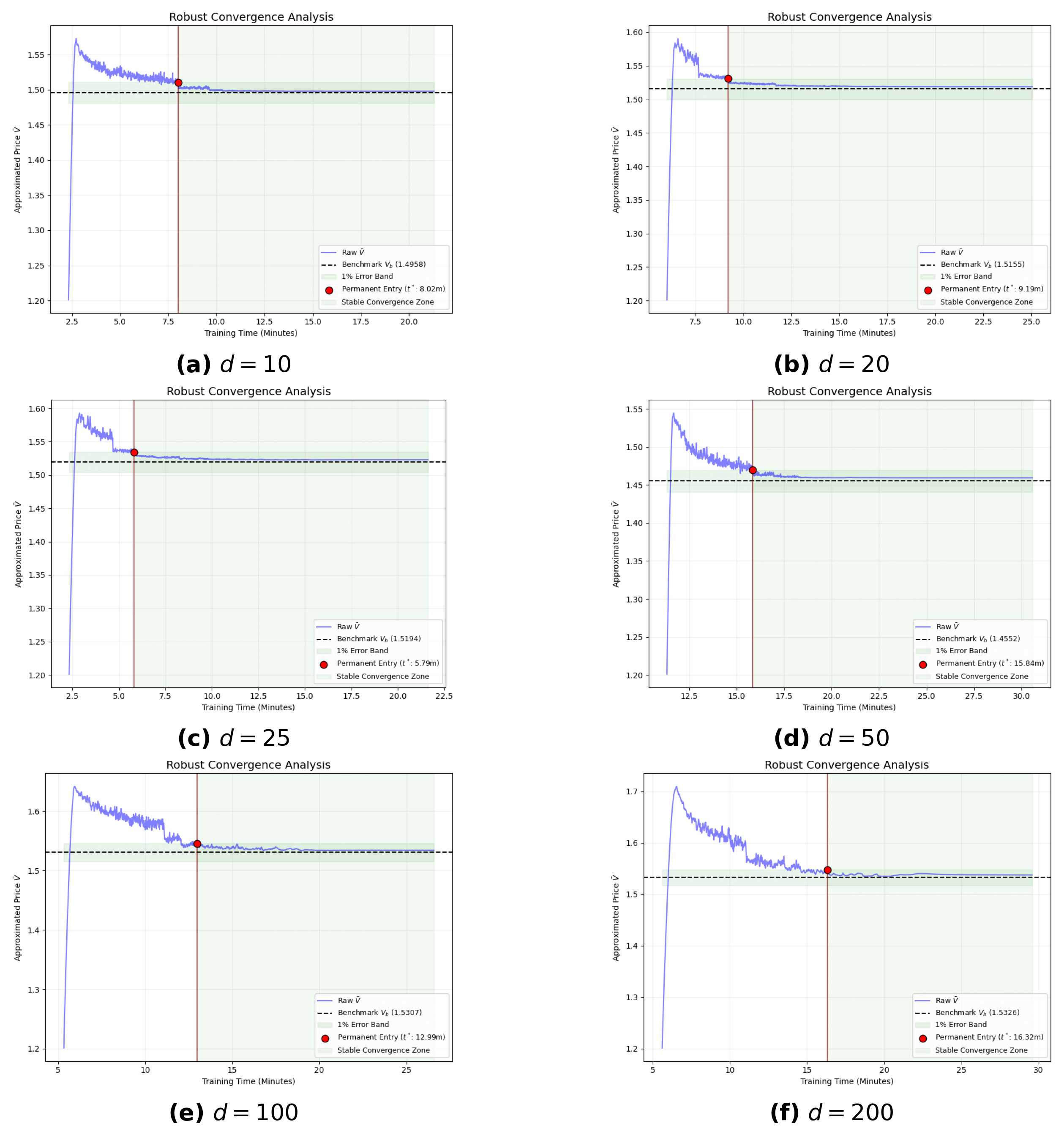}
    \caption{Robust convergence analysis of the DPM solver across dimensions $d \in \{10, 20, 25, 50, 100, 200\}$. \small\textit{The red marker indicates the point of permanent entry ($t^*$) into the 1\% error tolerance band (shaded green). Numerical experiments were performed on an NVIDIA G4 GPU with the following parameters. Problem parameters:  $r = \mu = 0.05$, $\sigma = \sqrt{2}$, $K = 2$, $T = 1$. Penalty and discretization: $N = 99$ time steps, penalty $\lambda = \frac{1}{\sqrt{T/N}}=9.95$. Training:  $30,000$ epochs, batch size $8,192$ (single-batch training per epoch).}}
    \label{fig:global_convergence_grid}
\end{figure}
\subsection{Comparative analysis of loss function robustness: MSE v.s. \texorpdfstring{$L^1$}{L1}}
The theoretical error analysis suggests that the $L^1$ loss function is a natural candidate for the DPM framework. Given that the majority of Deep BSDE solvers traditionally utilize MSE, we conduct a comparative study to evaluate the efficacy of the $L^1$ loss function employed in this work. Figure \ref{fig:combined_metrics} illustrates three critical performance metrics: total training time, the convergence stability threshold ($t^*$) required to reach a stable $1\%$ error, and the final relative pricing error. While numerical results exhibit minor variations across different dimensions, both loss functions consistently perform within the same order of magnitude. These marginal fluctuations likely originate from stochastic hardware latencies, such as CPU-GPU communication overhead, rather than inherent algorithmic disparities. Consequently, these results suggest that the DPM framework is numerically indifferent to the choice between MSE and $L^1$ loss functions, thereby validating the robustness of the choice of the loss functions within our approach.
\begin{figure}[htbp]
    \centering
    \includegraphics[width=\textwidth]{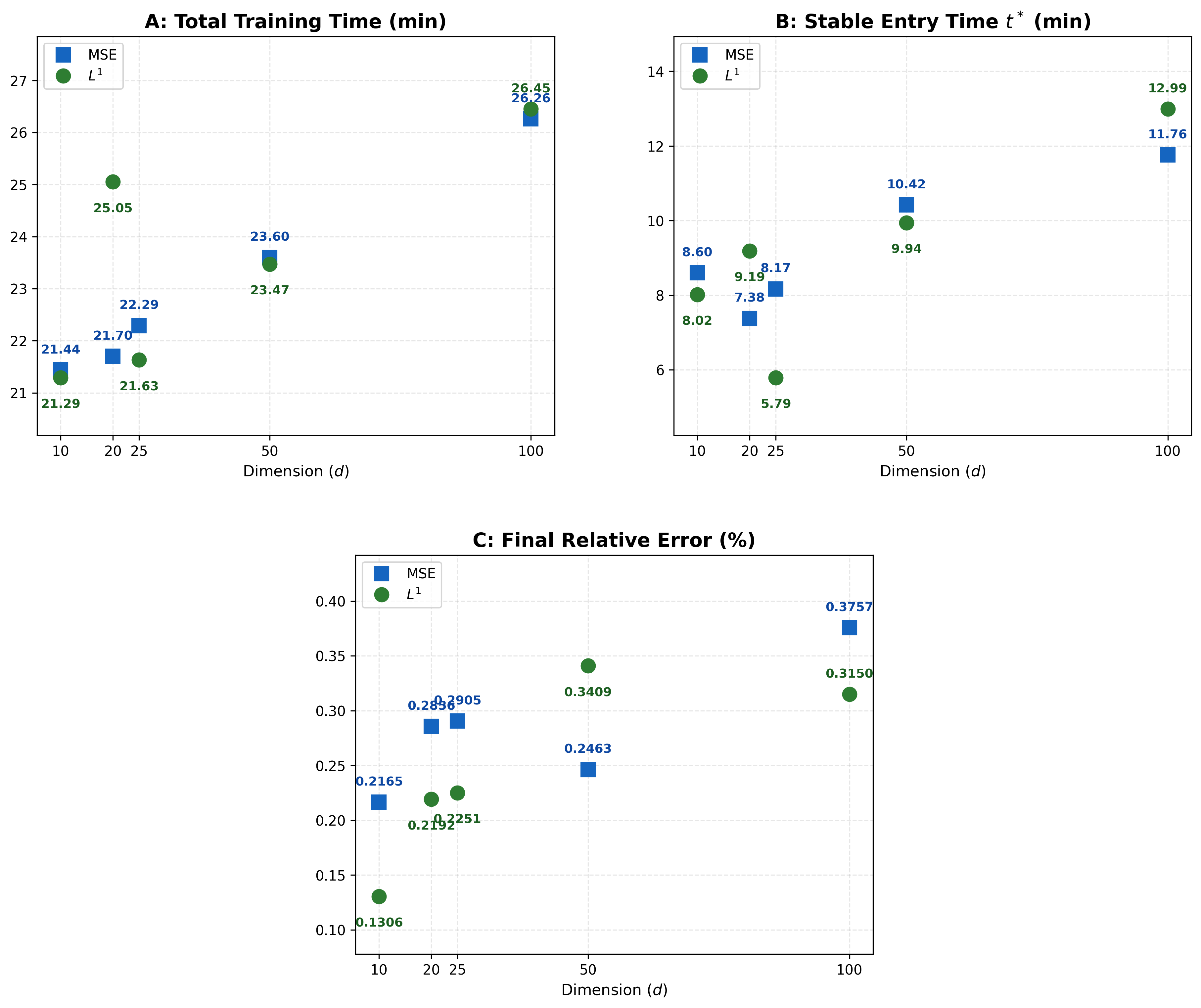}
    \caption{Numerical performance comparison between MSE and $L^1$ Loss functions: \small\textit{(a) Total training time for 30,000 epochs; (b) Time required to reach the stable 1\% error band ($t^*$); (c) Relative error of the approximated solution $V$ at the end of training, calculated by $\frac{|\bar{V}-V_{b}|}{V_b}$. Numerical experiments were performed on an NVIDIA G4 GPU with the following parameters. Problem parameters:  $r = \mu = 0.05$, $\sigma = \sqrt{2}$, $K = 2$, $T = 1$. Penalty and discretization: $N = 99$ time steps, penalty $\lambda = \frac{1}{\sqrt{T/N}}=9.95$. Training:  $30,000$ epochs, batch size $8,192$ (single-batch training per epoch).}}
    \label{fig:combined_metrics}
\end{figure}

\section{Conclusions}\label{sec:conclusion}
We have introduced the DPM for solving high dimensional optimal stopping problems in a continuous time setting. The DPM, inspired by the penalty method, utilizes the Deep BSDE technique to approximate solutions to variational inequalities. We have demonstrated that the error of the DPM can be bounded by the cost function and the square root of the time step size. The error analysis underscores the importance of selecting the penalization parameter. Moreover, numerical tests have shown that the DPM works efficiently and accurately in high dimensional models. We believe that the DPM can be extended to apply to optimal switching models, given the validity of the penalty method in solving variational inequality systems. We have provided a detailed error analysis for a general optimal switching model in Appendix B.  However, further investigation is required to assess the computational efficiency of the algorithm.

\appendix

\section{Proofs}
\subsection{Proof of Proposition \ref{Prop:penaltyTerm}}

Let $u(t,x) = e^{at}V^{\lambda}(t,x)$, where $a$ is to be chosen later, then it follows from some algebra that
\begin{align*}
\mathcal{L}u(t,x)-(r+a)u(t,x)+f_0(t,x)+\lambda (p_0(t,x)-u(t,x))^+ = 0\quad (t,x)\in[0,T)\times \mathbb{R}^d,
\end{align*}
with the terminal condition
\begin{align*}
u(T,x) = h_0(x),\quad x\in\mathbb{R}^d,
\end{align*}
where $p_0(t,x) = e^{at}p(t,x), h_0(x) = e^{aT}h(x), f_0(t,x) = e^{at}f(t,x).$\par 
Define $Q(x) = \vert x\vert^{2k}+d',$ where $k>2, d'>0$, then some algebra yields that
\begin{empheq}[left=\empheqlbrace]{align}\label{inequality:e1}
  b_i(t,x)\frac{\partial Q}{\partial x_i}(x) &= 2kb_i(t,x)\vert x\vert^{2k-2}x_i,\quad i = 1,2,\dots,d,\nonumber\\
a_{i,j}(t,x)\frac{\partial^2 Q}{\partial x_i\partial x_j}(x) &= 4k(k-1)a_{i,j}(t,x) \vert x\vert^{2k-4} x_ix_j, \quad i\neq j, i,j =1,2,\dots,d,\\
a_{i,i}(t,x)\frac{\partial^2 Q}{\partial x_i^2}(x) &= a_{i,i}(t,x)(4k(k-1)\vert x\vert^{2k-4} x_i^2 + 2k\vert x\vert^{2k-2}),\quad i = 1,2,\dots,d.\nonumber
\end{empheq}
Note that $b_i(t,x), \sigma_{i,j}(t,x)$ satisfy (\ref{Lip}), $b_{i}$ has linear growth in $x$ and $a_{ij}(t,x)$ has quadratic growth in $x$, i.e., there exist positive constants $c_1,d_1, c_2, d_2$ such that
\begin{empheq}[left=\empheqlbrace]{align}
  \vert b_i(t,x) \vert  &< c_1\vert x\vert+d_1, \quad i = 1,2,\dots,d,\nonumber\\
 \vert a_{ij}(t,x)\vert &< c_2\vert x\vert^2+d_2, \quad  i,j =1,2,\dots,d.\nonumber
\end{empheq}
Thus it is easy to see that $\mathcal{L}Q(x)$ has polynomial growth with order $2k$. Indeed, combining the above inequalities with (\ref{inequality:e1}), we then have
\begin{empheq}[left=\empheqlbrace]{align}\label{inequality:e2}
\bigg | b_i(t,x)\frac{\partial Q}{\partial x_i}(x)\bigg | &< 2kc_1\vert x\vert^{2k} + 2kd_1\vert x\vert^{2k-1} ,\quad i = 1,2,\dots,d,\nonumber\\
\bigg | a_{i,j}(t,x)\frac{\partial^2 Q}{\partial x_i\partial x_j}(x) \bigg | &< 4k(k-1)c_2\vert x\vert^{2k}+4k(k-1)d_2\vert x\vert^{2k-2} , \quad i\neq j, i,j =1,2,\dots,d,\\
\bigg | a_{i,i}(t,x)\frac{\partial^2 Q}{\partial x_i^2}(x)\bigg | &< (4k(k-1)c_2+2kc_2)\vert x\vert^{2k}+(4k(k-1)d_2+2kd_2)\vert x\vert^{2k-2},\ i = 1,2,\dots,d.\nonumber
\end{empheq}
Therefore, there exist constants $a>0,d'>0,$ such that
\begin{align}\label{inequality:e3}
\mathcal{L}Q(x) < (r+a)Q(x).
\end{align}
Let $w(t,x) = u(t,x) + \epsilon Q(x), \epsilon>0,$ then it is to see that
\begin{align*}
\mathcal{L}w(t,x) - (r+a)w &= \mathcal{L}u(t,x)-(r+a)u(t,x)+\epsilon(\mathcal{L}Q(x)-(r+a)Q(x))\\&=-\lambda (p_0(t,x)-u(t,x))^+-f_0(t,x) +\epsilon(\mathcal{L}Q(x)-(r+a)Q(x))\\&\le -\lambda (p_0(t,x)-w(t,x))^+-f_0(t,x) +\epsilon(\mathcal{L}Q(x)-(r+a)Q(x)).
\end{align*}
Then the above inequality and (\ref{inequality:e3}) result in
\begin{align*}
\lambda (p_0(t,x)-w(t,x))^+\le -(\mathcal{L}w(t,x) -(r+a)w(t,x)+f_0(t,x)).
\end{align*}
Let $l(t,x) = p_0(t,x)-w(t,x),$ then the above inequality yields that
\begin{align}\label{inequality:e4}
\lambda l(t,x)^+ \le \mathcal{L}l(t,x) - (r+a)l(t,x) -(\mathcal{L}p_0(t,x)-(r+a)p_0(t,x) +f_0(t,x)).
\end{align}
Note that $p_0$ has linear growth in $x$ and $w$ has polynomial growth with order $k>2$, then we have that there exists $M>0,$ such that $l(t,x)<0,$ whenever $\vert x\vert >M$. Also, it is easy to see that 
\begin{align*}
l(T,x)&= e^{aT}(p(T,x)-h(x)) -\epsilon Q(x)\\&= e^{aT}(p(T,x)-\max\{g(x),p(T,x)\})-\epsilon Q(x)\\&\le 0. 
\end{align*}
Suppose $l$ takes the positive maximum at $(t_0,x_0),$ if it exists, then $(t_0,x_0)\in[0,T)\times (-M,M)^m.$ Thus, it is easy to see that $ \mathcal{L}l(t_0,x_0) - (r+a)l(t_0,x_0)\le 0,$ which, together with (\ref{inequality:e4}), yields that
\begin{align}\label{inequality:e5}
\lambda l(t,x)^+ &\le  -(\mathcal{L}p(t_0,x_0)-(r+a)p(t_0,x_0) +f_0(t_0,x_0))\nonumber\\&\le -\inf_{(t,x)\in[0,T)\times\mathbb{R}^d}(\mathcal{L}p_0(t,x)-(r+a)p_0(t,x) +f_0(t,x))\nonumber\\& =-e^{at}\inf_{(t,x)\in[0,T)\times\mathbb{R}^d}(\mathcal{L}p(t,x)-rp(t,x) +f(t,x)) ,
\end{align}
for all $(t,x)\in[0,T)\times\mathbb{R}^d.$\par 
Note that 
\begin{align*}
\lambda (p(t,x)-V^{\lambda}(t,x))^+&=e^{-at}\lambda (p_0(t,x)-u(t,x))^+\\&=e^{-at}\lambda (p_0(t,x)-w(t,x)+\epsilon Q(x))^+\\&=e^{-at}\lambda (l(t,x)+\epsilon Q(x))^+\\&\le e^{-at}\lambda l(t,x)^++e^{-at}\lambda\epsilon Q(x).
\end{align*}
Then it follows from  (\ref{inequality:e5}) that
\begin{align*}
\lambda (p(t,x)-V^{\lambda}(t,x))^+\le -\inf_{(t,x)\in[0,T)\times\mathbb{R}^d}(\mathcal{L}p(t,x)-rp(t,x) +f(t,x)) +e^{-at}\lambda\epsilon Q(x).
\end{align*}
Finally, let $\epsilon\to 0,$ then we obtain
\begin{align*}
\lambda (p(t,x)-V^{\lambda}(t,x))^+\le -\inf_{(t,x)\in[0,T)\times\mathbb{R}^d}(\mathcal{L}p(t,x)-rp(t,x) +f(t,x)).
\end{align*}
This complete the proof. 
\subsection{Proof of Theorem \ref{Prop:penaltyError}}

First, it follows from the PDE problems (\ref{FBPDE}) and (\ref{PPDE}) that
\begin{align*}
\mathcal{L}( V(t,x)-V^{\lambda}(t,x))-r (V(t,x)-V^{\lambda}(t,x))\le \lambda (p(t,x)-V^{\lambda}(t,x))^+.
\end{align*}
Note that $\lambda (p(t,x)-V(t,x))^+ = 0$, then we have
\begin{align*}
\mathcal{L}( V(t,x)-V^{\lambda}(t,x))-r (V(t,x)-V^{\lambda}(t,x))\le\lambda (p(t,x)-V^{\lambda}(t,x))^+-\lambda (p(t,x)-V(t,x))^+  ,
\end{align*}
which yields 
\begin{align*}
\mathcal{L}( V(t,x)-V^{\lambda}(t,x))-(r-\beta(t,x)) (V(t,x)-V^{\lambda}(t,x))\le0,
\end{align*}
where $\beta(t,x) =\frac{\lambda (p(t,x)-V(t,x))^+  -\lambda (p(t,x)-V^{\lambda}(t,x))^+}{V(t,x)-V^{\lambda}(t,x)}1_{\{V(t,x)-V^{\lambda}(t,x)\neq 0\}}.$\par 
Note that $\beta(t,x)\le 0, \forall (t,x)\in [0.T)\times \mathbb{R}^d,$ then it follows from the comparison principle and $V(T,x) = V^{\lambda}(T,x)$ that $V(t,x)\ge V^{\lambda}(t,x), \forall (t,x)\in [0.T)\times \mathbb{R}^d.$\par 
Second, in the stopping region $\mathcal{S}$, we have that
\begin{align*}
V(t,x)-V^{\lambda}(t,x) = p(t,x)-V^{\lambda}(t,x)\le (p(t,x)-V^{\lambda}(t,x))^+,
\end{align*}
then Proposition \ref{Prop:penaltyTerm} yields that
\begin{align}\label{inequality:e6}
 \sup_{(t,x)\in\mathcal{S}}(V(t,x)-V^{\lambda}(t,x) )\le  \frac{C}{\lambda}.
\end{align}
In the continuation region $\mathcal{C}$, we have that
\begin{align*}
\mathcal{L}(V(t,x)-V^{\lambda}(t,x)) = r(V(t,x)-V^{\lambda}(t,x))+\lambda (p(t,x)-V^{\lambda}(t,x))^+\ge 0.
\end{align*}
Therefore, it follows from the maximum principle that the positive maximum of $V(t,x)-V^{\lambda}(t,x)$ is taken at the boundary of the continuation region, which, together with (\ref{inequality:e6}), gives that
\begin{align*}
\sup_{(t,x)\in\mathcal{C}}(V(t,x)-V^{\lambda}(t,x))\le \sup_{(t,x)\in\mathcal{S}}(V(t,x)-V^{\lambda}(t,x))\le  \frac{C}{\lambda}.
\end{align*}
This completes the proof.
\subsection{Proof of Theorem \ref{Thm:discretization}}
It follows from the BSDE (\ref{BSDE}) that
\begin{align*}
Y^{\lambda}_{t_i} &=  \mathbb{E}_{t_i}[Y^{\lambda}_{t_{i+1}}] +\mathbb{E}_{t_i}\left[\int^{t_{i+1}}_{t_i}f_1(s,X_s)ds +\lambda \int^{t_{i+1}}_{t_i}(-Y^{\lambda}_{s})^+ds\right]\\
&= \mathbb{E}_{t_i}[Y^{\lambda}_{t_{i+1}}] + (f_1(t_i,X_{t_i}) + \lambda (-Y^{\lambda}_{t_i})^+)h + e_i,
\end{align*}
where 
\begin{align*}
e_i = \mathbb{E}_{t_{i}}\left[\int^{t_{i+1}}_{t_i} (f_1(s,X_s)-f_1(t_i,X_{t_i}))ds+ \lambda \int ^{t_{i+1}}_{t_i}(     (-Y^{\lambda}_{s})^+-(-Y^{\lambda}_{t_i})^+)ds\right].
\end{align*}
For the ease of exposition, we define $\delta Y^{\lambda}_{t_i}: = Y^{\lambda}_{t_i} -Y^{\lambda,\pi}_{t_i}, \delta f_{t_i} := f_1(t_i,X_{t_i}) -f_1(t_i,X^{\pi}_{t_i}).$ Then, we have that
\begin{align*}
\delta Y^{\lambda}_{t_i} = \mathbb{E}_{t_i}[\delta Y^{\lambda}_{t_{i+1}}]+ \delta f_{t_i}h + \lambda ( (-Y^{\lambda}_{t_i})^+-(-Y^{\lambda,\pi}_{t_i})^+)+ e_i.
\end{align*}
Moving $\lambda ( (-Y^{\lambda}_{t_i})^+-(-Y^{\lambda,\pi}_{t_i})^+)$ to the left side, we have that
\begin{align*}
(1-\eta_{t_i})\delta Y^{\lambda}_{t_i} = \mathbb{E}_{t_i}[\delta Y^{\lambda}_{t_{i+1}}]+ \delta f_{t_i}h + e_i,
\end{align*}
where $\eta_{t_i} = \lambda\frac{ (-Y^{\lambda}_{t_i})^+-(-Y^{\lambda,\pi}_{t_i})^+}{\delta Y^{\lambda}_{t_i}}1_{\{\delta Y^{\lambda}_{t_i}\neq 0\}}$.\par 
As $\eta_{t_i} \le 0$, then
\begin{align*}
\vert\delta Y^{\lambda}_{t_i}\vert \le \mathbb{E}_{t_i}\vert\delta Y^{\lambda}_{t_{i+1}}\vert+ \vert\delta f_{t_i}\vert h+ \vert e_i\vert
\end{align*}
Taking the expectation on the both sides, we have that
\begin{align*}
\mathbb{E}\vert\delta Y^{\lambda}_{t_i}\vert \le \mathbb{E}\vert\delta Y^{\lambda}_{t_{i+1}}\vert+ \mathbb{E}\vert\delta f_{t_i}\vert h +\mathbb{E} \vert e_i\vert,
\end{align*}
which gives rise to
\begin{align*}
\sum_{j=i}^{N-1}\mathbb{E}\vert\delta Y^{\lambda}_{t_j}\vert \le \sum_{j=i}^{N-1}\mathbb{E}\vert\delta Y^{\lambda}_{t_{j+1}}\vert+ h\sum_{j=i}^{N-1}\mathbb{E}\vert\delta f_{t_j}\vert +\sum_{j=i}^{N-1}\mathbb{E} \vert e_j\vert,
\end{align*}
which yields that
\begin{align}\label{Estimate}
\mathbb{E}\vert\delta Y^{\lambda}_{t_i}\vert \le \mathbb{E}\vert\delta Y^{\lambda}_{T}\vert+ h\sum_{j=i}^{N-1}\mathbb{E}\vert\delta f_{t_j}\vert +\sum_{j=i}^{N-1}\mathbb{E} \vert e_j\vert.
\end{align}
First, we estimate $ \mathbb{E}\vert\delta Y^{\lambda}_{T}\vert.$ It is easy to see from Lemma \ref{Lem:1} that
\begin{align}\label{Estimate:terminal}
\mathbb{E}|\delta Y^{\lambda}_T|  = \mathbb{E}|h_1(X_T) - h_1(X_T^{\pi})| \le K_1\sqrt{h}.
\end{align}
Second, regarding $h\sum\limits_{j=i}^{N-1}\mathbb{E}\vert\delta f_{t_j}\vert$, following Lemma \ref{Lem:1}, we have that
\begin{align}\label{Estimate:running}
h\sum_{j=i}^{N-1}\mathbb{E}\vert\delta f_{t_j}\vert \le h\sum_{j=i}^{N-1}K_1\sqrt{h} \le h K_1\sqrt{h}N=  K_1 T\sqrt{h}. 
\end{align}
Finally, let us focus on $\sum\limits_{j=i}^{N-1}\mathbb{E} \vert e_j\vert$. Note that
\begin{align}\label{Estimate:error}
\sum_{j=i}^{N-1}\mathbb{E} \vert e_j\vert&\le \sum_{i=0}^{N-1}\mathbb{E}\left[\left|\int^{t_{i+1}}_{t_i} (f_1(s,X_s)-f_1(t_i,X_{t_i}))ds\right|+ \lambda \left|\int ^{t_{i+1}}_{t_i}(     (-Y^{\lambda}_{s})^+-(-Y^{\lambda}_{t_i})^+)ds\right|\right].
\end{align}
Lemma \ref{Lem:1} yields that
\begin{align*}
\mathbb{E}\left|\mathbb{E}_{t_{i}}\left[\int^{t_{i+1}}_{t_i} (f_1(s,X_s)-f_1(t_i,X_{t_i}))ds\right]\right|&\le K_1\int^{t_{i+1}}_{t_i} \sqrt{s-t_i}ds\le K_1\sqrt{h}h,
\end{align*}
which gives that
\begin{align}\label{Estimate:runningerror}
\sum_{i=0}^{N-1}\mathbb{E}\left|\int^{t_{i+1}}_{t_i} (f_1(s,X_s)-f_1(t_i,X_{t_i}))ds\right|\le K_1T\sqrt{h}.
\end{align}
Using It\^{o}-Tanaka's formula (\cite{karatzas2012brownian}), we have that
\begin{align*}
(-Y^{\lambda}_s)^+ - (-Y^{\lambda}_{t_i})^+ = \int^s_{t_i}-1_{\{Y^{\lambda}_u\le 0\}}dY^{\lambda}_u + \frac{1}{2}(L^0(s)-L^0(t_i)),
\end{align*}
where $L^0(\cdot)$ is the local time of $Y^{\lambda}$ at $0$.\par 
Then we obtain that
\begin{align*}
 &\left|\mathbb{E}_{t_{i}}\left[ \int ^{t_{i+1}}_{t_i}( (-Y^{\lambda}_{s})^+-(-Y^{\lambda}_{t_i})^+)ds \right]\right |\\&\le \int ^{t_{i+1}}_{t_i} \left|\mathbb{E}_{t_{i}}\left[\int^s_{t_i}-1_{Y^{\lambda}_u\le 0}dY^{\lambda}_u + \frac{1}{2}(L^0(s)-L^0(t_i))\right]\right |ds\\&=\int ^{t_{i+1}}_{t_i} \left|\mathbb{E}_{t_{i}}\left[\int^s_{t_i}1_{Y^{\lambda}_u\le 0}(f(u,X_u)+\lambda (-Y^{\lambda}_u)^+)du + \frac{1}{2}(L^0(s)-L^0(t_i))\right]\right |ds\\&\le I_{1i}+I_{2i}+I_{3i},
\end{align*}
where $I_{1i} = \mathbb{E}_{t_{i}}\int ^{t_{i+1}}_{t_i} \int^s_{t_i}|f(u,X_u)|duds, I_{2i}=\mathbb{E}_{t_{i}}\int ^{t_{i+1}}_{t_i} \int^s_{t_i}\lambda (-Y^{\lambda}_u)^+du ds, I_{3i} = \mathbb{E}_{t_{i}}\int ^{t_{i+1}}_{t_i} \frac{1}{2}(L^0(s)-L^0(t_i))ds.$\par 
Then (\ref{Estimate:error}) and (\ref{Estimate:runningerror}) yield that
\begin{align}\label{Estimate:error1}
\sum_{j=i}^{N-1}\mathbb{E} \vert e_j\vert\le K_1T\sqrt{h} + \lambda\left(\sum_{i=0}^{N-1}\mathbb{E}I_{1i}+  \sum_{i=0}^{N-1}\mathbb{E}I_{2i}+  \sum_{i=0}^{N-1}\mathbb{E}I_{3i}\right).
\end{align}
Moreover, it follows from Lemma \ref{Lem:1} that
\begin{align*}
 \lambda\sum_{i=0}^{N-1}\mathbb{E}I_{1i} \le\lambda N K_1 h^2 = K_1T\lambda h.
\end{align*}
Also, Proposition \ref{Prop:penaltyTerm} gives that
\begin{align*}
 \lambda\sum_{i=0}^{N-1}\mathbb{E}I_{2i}\le \lambda NCh^2 = \lambda CTh.
\end{align*}
And following the definition of $I_{3i}$, we easily obtain that
\begin{align*}
\lambda\sum_{i=0}^{N-1}\mathbb{E}I_{3i}\le \frac{1}{2}\lambda\sum_{i=0}^{N-1}\mathbb{E}\left[h(L^0(t_{i+1})-L^0({t_i}))\right] \le \frac{1}{2}\lambda h\mathbb{E} L^0(T).
\end{align*}
Plug the three above inequalities into (\ref{Estimate:error1}), we then obtain the estimate for  $\sum\limits_{j=i}^{N-1}\mathbb{E} \vert e_j\vert$,
\begin{align}\label{Estimate:error2}
\sum_{j=i}^{N-1}\mathbb{E} \vert e_j\vert\le K_1T\sqrt{h} + (K_1T+CT+\frac{1}{2} \mathbb{E} L^0(T))\lambda h.
\end{align}
Plug (\ref{Estimate:terminal}), (\ref{Estimate:running}) and (\ref{Estimate:error2}) into (\ref{Estimate}), then we have the conclusion and complete the proof.
\subsection{Proof of Corollary \ref{Cor:1}}
Note that $V^{\lambda}(t_i,X_{t_i}) =e^{rt_i}Y^{\lambda}_{t_i} +p(t_i,X_{t_i})$, then we have that
\begin{align*}
\mathbb{E}\vert e^{rt_i}Y^{\lambda,\pi}_{t_i} +p(t_i, X_{t_i})- V(t_i,X_{t_i})|\le |V(t_i,X_{t_i})-V^{\lambda}(t_i,X_{t_i})|+\mathbb{E}\vert Y^{\lambda,\pi}_{t_i}  - Y^{\lambda}_{t_i} |e^{rt_i}.
\end{align*}
It follows from Theorem \ref{Prop:penaltyError} and Theorem \ref{Thm:discretization} that
\begin{align*}
\mathbb{E}\vert e^{rt_i}Y^{\lambda,\pi}_{t_i} +p(t_i, X_{t_i})- V(t_i,X_{t_i})| \le \frac{C}{\lambda}+ C_0 (\sqrt{h}+\lambda h)e^{rt_i}.
\end{align*}
 Then let $\lambda=\frac{1}{\sqrt{h}}$, we obtain the conclusion and complete the proof.
\subsection{Proof of Lemma \ref{Lem:3}}
It follows from (\ref{DDBSDE}) that 
\begin{align*}
Y^{\lambda,\pi,k}_{t_i} = \mathbb{E}_{t_i}[Y^{\lambda,\pi,k}_{t_{i+1}}] + (f_1(t_i,X^{\pi}_{t_i}) + \lambda (-Y^{\lambda,\pi,k}_{t_i})^+)h,\quad \forall i = 0,1,2,\dots,N-1, k=1,2,
\end{align*}
which yields that
\begin{align*}
\delta Y^{\lambda,\pi}_{t_i} = \mathbb{E}_{t_i}[\delta Y^{\lambda,\pi}_{t_{i+1}}] +( \lambda (-Y^{\lambda,\pi,1}_{t_i})^+- \lambda (-Y^{\lambda,\pi,2}_{t_i})^+)h = \mathbb{E}_{t_i}[\delta Y^{\lambda,\pi}_{t_{i+1}}] +\eta^{\pi}_{t_i} h \delta Y^{\lambda,\pi}_{t_i},
\end{align*}
where  $\eta^{\pi}_{t_i} = \lambda\frac{ (-Y^{\lambda,\pi,1}_{t_i})^+-(-Y^{\lambda,\pi,2}_{t_i})^+}{\delta Y^{\lambda,\pi}_{t_i}}1_{\{\delta Y^{\lambda,\pi}_{t_i}\neq 0\}}$.\par 
Then we have that
\begin{align*}
    (1-\eta^{\pi}_{t_i} h )\delta Y^{\lambda,\pi}_{t_i}  = \mathbb{E}_{t_i}[\delta Y^{\lambda,\pi}_{t_{i+1}}] 
\end{align*}
As $\eta^{\pi}_{t_i} \le 0,$ we then have that
\begin{align*}
     |\delta Y^{\lambda,\pi}_{t_i}| \le (1-\eta^{\pi}_{t_i} h )|\delta Y^{\lambda,\pi}_{t_i}|  \le \mathbb{E}_{t_i}|\delta Y^{\lambda,\pi}_{t_{i+1}}|, 
\end{align*}
which yields
\begin{align*}
\mathbb{E}|\delta Y^{\lambda,\pi}_{t_i}| \le \mathbb{E}|\delta Y^{\lambda,\pi}_{t_{i+1}}|, 0\le i\le N-1.
\end{align*}
This completes the proof.

\subsection{Proof of Theorem \ref{Thm:errorBounds}}
Note that
\begin{align*}
\mathbb{E}\vert \mathcal{V}^{\lambda,\pi,\theta}_{t_{i}} - Y^{\lambda}_{t_i}\vert\le \mathbb{E}\vert \mathcal{V}^{\lambda,\pi,\theta}_{t_{i}} - Y^{\lambda,\pi}_{t_i}\vert+\mathbb{E}\vert Y^{\lambda,\pi}_{t_i}- Y^{\lambda}_{t_i}\vert,
\end{align*}
then the theorem follows from Lemma \ref{Lem:3} and Theorem \ref{Thm:discretization}.

 \subsection{Proof of Corollary \ref{corollary:3}}
Note that $Y^{\lambda}_{t_i} = U^{\lambda}(t_i,X_{t_i}) = e^{-rt}(V^{\lambda}(t_i,X_{t_i})-p(t_i,X_{t_i})),$ and 
\begin{align*}
\mathbb{E}\vert \mathcal{V}^{\lambda,\pi,\theta}_{t_{i}}e^{rt_i}+p(t_i,X_{t_i}) - V({t_i},X_{t_i})\vert&\le \mathbb{E}\vert \mathcal{V}^{\lambda,\pi,\theta}_{t_{i}}e^{rt_i}+p(t_i,X_{t_i}) - V^{\lambda}({t_i},X_{t_i})\vert+\mathbb{E}\vert  V^{\lambda}({t_i},X_{t_i})- V({t_i},X_{t_i})\vert\\&=e^{rt_i}\mathbb{E}\vert \mathcal{V}^{\lambda,\pi,\theta}_{t_{i}} - Y^{\lambda}_{t_i}\vert+\mathbb{E}\vert  V^{\lambda}({t_i},X_{t_i})- V({t_i},X_{t_i})\vert,
\end{align*}
Then it follows from Theorem \ref{Thm:errorBounds} and Proposition \ref{Prop:penaltyError} that
 \begin{align*}
\sup_{0\le i\le N}\mathbb{E}\vert \mathcal{V}^{\lambda,\pi,\theta}_{t_{i}}e^{rt_i}+p(t_i,X_{t_i}) - V({t_i},X_{t_i})\vert\le C (\sqrt{h}+\lambda h + \frac{1}{\lambda}) +  e^{rT}\mathbb{E}\vert \mathcal{V}^{\lambda,\pi,\theta}_{t_N}-h_1(X_{t_N})\vert,
 \end{align*}
 where $C$ is a constant depending on $T,X_0$.\par 
 Let $\lambda = \frac{1}{\sqrt{h}},$ then the corollary is an immediate result and the proof is complete.

\section{Extension to Optimal Switching}
The key advantage of the DPM lies in eliminating the backward iterations and repeated optimizations inherent to traditional dynamic programming, thereby capturing the dynamics of a stopping model in a single step. Naturally, this framework can be extended to solving optimal switching problems. In this section, we briefly introduce the DPM for a general optimal switching problem and adapt the DPM's error analysis from optimal stopping models to the optimal switching setting.

\subsection{The optimal switching problem}
A switching strategy is an impulse control $c=(\tau_n, \alpha_n)_{n \ge 1}$, consisting of a double sequence indexed by $n \in \mathbb{N}^* = \mathbb{N} \setminus \{0\}$, where:
\begin{itemize}
    \item $\tau_n \in \mathcal{T}$ are $\{\mathcal{F}_t\}_{t\ge 0}$-stopping times in $[0, \infty]$ satisfying $\tau_n < \tau_{n+1}$ and $\tau_n \to \infty$ a.s., representing the decisions on \textit{``when to switch''};
    \item $\alpha_n$ are $\mathcal{F}_{\tau_n}$-measurable random variables taking values in $\Lambda$, representing the new value of the regime at time $\tau_n$ (held until time $\tau_{n+1}$), or the decisions on \textit{``where to switch''}.
\end{itemize}
We denote by $\mathcal{A}$ the set of all such admissible impulse controls.
An admissible impulse control determines a switching process $a = (a_t)_{0 \leq t \leq T}$ defined by
$\mathcal{F}_{\tau_n}$-measurable $\Lambda$-valued random variables $\{\alpha_n\}_{n \geq 0}$,
such that:
\[
a^i_t = \sum_{n \geq 1} \alpha_{n-1} \mathbf{1}_{(\tau_{n-1},\tau_n}(t), t\ge 0, a^i_{0-} = i
\]
with $\tau_M = T$ a.s. for some random variable with integer value $M \in L^2(\mathcal{F}_T)$.\par 
We suppose that switching between different regimes incurs instantaneous cost defined by
\[
k(i,j) \geq 0 \quad \forall i,j \in \Lambda,
\]
satisfying:
\begin{enumerate}
    \item $k(i,j) > 0$ for $i \neq j$ (strictly positive switching cost),
    \item $k(i,j) + k(j,l) \geq k(i,l)$, $\forall i,j,l \in \Lambda$ (triangle inequality),
    \item $k(i,i) = 0$ (no profit to maintain regime).
\end{enumerate}
The objective functional at regime $i$ under the control $c \in \mathcal{A}$ is:
\begin{equation}\label{eq:obj-sw}
  J_i(t,x;c) = \mathbb{E}\left[\int_t^T f(s, X_s, a_s)ds + g(X_T, a_T) - \sum_{i=1}^M k(\alpha_{i-1}, \alpha_i)\vert X_t = x\right]  
\end{equation}
where:
\begin{itemize}
\item $X$ is given by (\ref{SDE}),
    \item $f: [0,T] \times \mathbb{R}^d \times \Lambda \rightarrow \mathbb{R}$ is the running profit function,
    \item $g: \mathbb{R}^d \times \Lambda \rightarrow \mathbb{R}$ is the terminal payoff.
\end{itemize}

Through this appendix, we suppose that $f,g$ in \eqref{eq:obj-sw} for each regime satisfy Assumption \ref{Assp:2}.
The value function for regime $i \in \Lambda$ is:
$$
V_i(t,x) = \sup_{c \in \mathcal{A}} J_i(t,x;c).
$$\par 
It follows from the standard literature (e.g., \cite{hu2010multi}) that the value function can be characterized by a system of  the following reflected backward stochastic differential equations (RBSDEs) as follows.
\begin{align} \label{eq:RBSDE}
\begin{cases}
Y^i_t = g(X_T, i) + \int_t^T f(s, X_s, i)ds +\int_t^T dK^i_s - \int_t^T Z^i_sdW_s, \\
Y^i_s \geq \max_{j \neq i} \left(Y^j_s - k(i,j)\right), \\
\int_0^T \left[Y^i_s - \max_{j \neq i} \left(Y^j_s - k(i,j)\right)\right] dK^i_s = 0,
\end{cases} 
\end{align}
where  $\{K^i_t\}_{0\le t\le T}$ is an increasing process.

The corresponding value function under the starting point $(t,x)$, $V_i: [0,T] \times \mathbb{R}^d \rightarrow \mathbb{R}$ satisfies the following system of quasi-variational inequalities:
\begin{align}\label{eq:HJB} 
\begin{cases}
\max\left\{\mathcal{L} V_i +f(t,x,i),  \max\limits_{j \neq i} (V_j(t,x) - k(i,j))-V_i(t,x)\right\}= 0, \\
V_i(T,x) = g(x,i), \quad \forall x \in \mathbb{R}^d,
\end{cases} 
\end{align}
where $\mathcal{L}$ is the differential operator given by (\ref{diffOperator}).
\subsection{The penalty approximation and the discretization error}
Consider the above optimal switching problem with $m$ regimes, where the value function is represented by the system of RBSDEs (\ref{eq:RBSDE}).
The penalty approximation system of (\ref{eq:RBSDE}) is given by:
\begin{align}\label{eq:penalized_bsde} 
Y^{\lambda,i}_t &= g(X_T, i) + \int_t^T \left[ f(s, X_s, i) + \lambda \Gamma^{\lambda,i}_s \right] ds - \int_t^T Z^{\lambda,i}_s d W_s,
\end{align}
where $\Gamma^{\lambda,i}_t = \sum\limits_{j =1}^{m} \left[ Y^{\lambda,j}_t - k(i,j)  - Y^{\lambda,i}_t \right]^+ .$\par 
 Similar to the main context, we employ the implicit time discretization scheme for the underlying dynamics (\ref{SDE}) and penalized BSDEs \eqref{eq:penalized_bsde}. For each $ k = 1,2,\dots,N-1$  and regime $i$, the discretization is: 
\begin{align} \label{eq:Y_approx}
Y_{t_k}^{\lambda,\pi,i} &= \mathbb{E}_{t_k} \left[ Y_{t_{k+1}}^{\lambda,\pi,i} \right] + h f(t_k, X_{t_k}^{\pi}, i) 
+\lambda h\Gamma^{\lambda,\pi,i}_{t_k}
\end{align} 
with terminal condition $Y^{\lambda,\pi,i}_{t_N} = g(X_T^{\pi},i)$, $X^{\pi}_{t_k}$ given by (\ref{DDynamics}) and $\Gamma^{\lambda,\pi,i}$ defined as follows
$$\Gamma^{\lambda,\pi,i}_{t_k} = \sum_{j =1}^{m} \left[ Y^{\lambda,\pi,j}_{t_k} - k(i,j)  - Y^{\lambda,\pi,i}_{t_k} \right]^+.$$
The $Z$-process is approximated by:
\begin{equation}\label{eq:Z_approx}
Z^{\lambda,\pi,i}_{t_k} = \frac{1}{h}\mathbb{E}_{t_k}\left[Y^{\lambda,\pi,i}_{t_{k+1}}(W_{t_{k+1}} - W_{t_k})\right].
\end{equation}
\begin{theorem}\label{thm:discrete_error_revised}
Suppose $b,\sigma$ in $X$ and $f(\cdot,\cdot,i),g(\cdot,i), i\in\Lambda$ in the objective functional follow Assumptions \ref{Assp:1} and \ref{Assp:2} respectively. Then we have
\begin{equation}
\sup_{i\in\Lambda} \sup_{0\le k\le N}  \mathbb{E}\left[ |Y^{\lambda,\pi,i}_{t_k} - Y^{i}_{t_k}| \right] \leq C(h^{\frac{1}{2}}+\lambda h +\lambda h^{\frac{1}{2}}+\lambda^{-1}). \label{eq:discrete_error_revised}
\end{equation}
Choosing $\lambda = h^{-\frac{1}{4}}$, the global error is:
\begin{equation}
\sup_{i\in\Lambda} \sup_{0\le k\le N}  \mathbb{E}\left[ |Y^{\lambda,\pi,i}_{t_k} - Y^{i}_{t_k}|\right] \leq C h^{\frac{1}{4}},\label{eq:global_error_revised}
\end{equation}
where $C$ is a constant depending on $T,X_0.$
\end{theorem}
\begin{proof}
First, by the triangle inequality, we have 
\begin{align}\label{eq:ineq}
\mathbb{E}\left[ |Y^{\lambda,\pi,i}_{t_k} - Y^{i}_{t_k}| \right]
\leq \mathbb{E}\left[ |Y^{\lambda,\pi,i}_{t_k} - Y^{\lambda,i}_{t_k}| \right]+\mathbb{E}\left[ |Y^{\lambda,i}_{t_k} - Y^{i}_{t_k}|\right].
\end{align}
Recall that
\begin{equation}\label{eq-220-1}
     Y_{t_k}^{\lambda, \pi, i}=\mathbb{E}_{t_k}[ Y_{t_{k+1}}^{\lambda, \pi, i}]+hf(t_k,X_{t_k}^{\pi},i)+h\lambda \Gamma^{\lambda,\pi,i}_{t_k},  
\end{equation}
\begin{equation} \label{eq-220-2}
     Y_{t_k}^{\lambda,  i}=Y_{t_{k+1}}^{\lambda,i}+\int_{t_{k}}^{t_{k+1}}[f(t,X_t,i)+\lambda \Gamma^{\lambda,i}_t]dt-\int_{t_{k}}^{t_{k+1}}Z^{\lambda,i}_tdW_t.
\end{equation} 
Subtracting (\ref{eq-220-2}) from (\ref{eq-220-1}) yields:
\[
\begin{array}{rcl}
 Y_{t_k}^{\lambda, \pi, i}-Y_{t_k}^{\lambda,  i}&=&\mathbb{E}_{t_k}[ Y_{t_{k+1}}^{\lambda, \pi, i}-Y_{t_{k+1}}^{\lambda,i}]+\mathbb{E}_{t_k}[\int_{t_{k}}^{t_{k+1}}[f(t_k,X_{t_k}^{\pi},i)-f(t,X_t,i)]dt\\
 &&+\lambda \mathbb{E}_{t_k}[\int_{t_{k}}^{t_{k+1}} [\Gamma^{\lambda,\pi,i}_{t_k}-\Gamma^{\lambda,i}_t]dt].  
\end{array}
\]
Taking expectations on both sides and applying the triangle inequality, we obtain
\begin{equation} \label{eq-inequality} 
\begin{array}{rcl}
\mathbb{E}[| Y_{t_k}^{\lambda, \pi, i}-Y_{t_k}^{\lambda,  i}|]&\leq &\mathbb{E}[ |Y_{t_{k+1}}^{\lambda, \pi, i}-Y_{t_{k+1}}^{\lambda,i}|]\\
&&+\mathbb{E}[\int_{t_{k}}^{t_{k+1}}|f(t,X_t,i)-f(t_k,X_{t_k}^{\pi},i)|dt\\
 &&+\lambda \mathbb{E}[\int_{t_{k}}^{t_{k+1}} |\Gamma^{\lambda,i}_t-\Gamma^{\lambda,\pi,i}_{t_k}|dt].
\end{array} 
\end{equation}

Next, we estimate the second term on the right-hand side of the inequality (\ref{eq-inequality}). 
Similar to Lemma \ref{Lem:1}, we have:
\begin{align*}
&\mathbb{E}\left[\int_{t_k}^{t_{k+1}}  |f(t, X_t, i)-f(t_k, X_{t_k}^{\pi}, i) | dt\right] \leq C h^{\frac{3}{2}}.
\end{align*}

Last, we estimate the third term on the right-hand side of the inequality (\ref{eq-inequality}).
\begin{equation}\label{eq-205-4}
    \begin{array}{ll}
      \mathbb{E}[|\Gamma^{\lambda,i}_t-\Gamma^{\lambda,\pi,i}_{t_k}|]  &  \\
        = \mathbb{E}[|\int_{t_{k}}^{t}d \Gamma^{\lambda,i}_s|]& \\
        \leq m \mathbb{E}[\int_{t_{k}}^{t}[|f(s,X_s,j)-f(s,X_s,i)|+\lambda |\Gamma^{\lambda,j}_s-\Gamma^{\lambda,i}_s|]ds]+\frac{1}{2}\mathbb{E}[L(t)-L(t_k)]&\\
        \leq m \mathbb{E}[\int_{t_{k}}^{t}[|f(s,X_s,j)-f(s,X_s,i)|+\lambda |\Gamma^{\lambda,j}_s-\Gamma^{\lambda,i}_s|]ds]+\frac{1}{2}\mathbb{E}[L(t_{k+1})-L(t_k)] &\\
        \leq C\sqrt{t-t_k}+\frac{1}{2}\mathbb{E}[L(t_{k+1})-L(t_k)]
    \end{array}
\end{equation}
where the $L^2$ boundedness of $\lambda |\Gamma^{\lambda,j}_s-\Gamma^{\lambda,i}_s|$ comes from the Lemma 2.1 in \cite{hu2010multi}, and $L(\cdot)$ is the local time of the $\{Y^{\lambda,j}_{t} - k(i,j)  - Y^{\lambda,i}_{t}\}=0$. Then, we have
\begin{equation*}
    \begin{array}{ll}
&\lambda  \mathbb{E}[\int_{t_{k}}^{t_{k+1}} |\Gamma^{\lambda,i}_t-\Gamma^{\lambda,\pi,i}_{t_k}|dt] \\
&=\lambda  \mathbb{E}[\int_{t_{k}}^{t_{k+1}} |\int_{t_{k}}^{t}d \Gamma^{\lambda,i}_s|dt] \\
&\leq \lambda\{ C \int_{t_{k}}^{t_{k+1}}\sqrt{t-t_k}dt+\frac{1}{2}h\mathbb{E}[(L(t_{k+1})-L(t_k))]\}\\
&\leq C(\lambda h^{\frac{3}{2}}+\frac{1}{2}\lambda h\mathbb{E}[L(t_{k+1})-L(t_k)]]).
    \end{array}
\end{equation*}
Taking summation of (\ref{eq-inequality}) from $k$ to $N$, and noting that $Y_{t_{N}}^{\lambda, \pi, i}=Y_{t_{N}}^{\lambda,i}$ by construction, we obtain:
\[
\mathbb{E}\left[ |Y^{\lambda,\pi,i}_{t_k} - Y^{\lambda,i}_{t_k}| \right]
\leq C(h^{\frac{1}{2}}+\lambda h^{\frac{1}{2}}+\lambda h ).
\]

Similar to Theorem \ref{Prop:penaltyError}, the second term in (\ref{eq:ineq}) becomes:
\[
\mathbb{E}\left[ |Y^{\lambda,i}_{t_k} - Y^i_{t_k}| \right] \leq \frac{C}{\lambda}, \quad  \forall i \in \Lambda.
\]
Finally, recall that  (\ref{eq:ineq}), we have
\begin{align*}
\mathbb{E}\left[ |Y^{\lambda,\pi,i}_{t_k} - Y^{i}_{t_k}| \right]
\leq C(h^{\frac{1}{2}}+\lambda h +\lambda h^{\frac{1}{2}}+\lambda^{-1}).
\end{align*}

Choosing $\lambda=h^{-\frac{1}{4}}$ yields the global error estimate: \begin{equation*}
\sup_{i\in\Lambda} \sup_{0\le k\le N}  \mathbb{E}\left[ |Y^{\lambda,\pi,i}_{t_k} - Y^{i}_{t_k}|\right] \leq C h^{\frac{1}{4}}.
\end{equation*}
This completes the proof.
\end{proof}
\subsection{Error bounds for the DPM}
In this subsection, we make an analysis of the error bound for the DPM. We implement the Deep BSDE framework developed in section \ref{sec:DPM} on the penalized BSDE system (\ref{eq:penalized_bsde}). In more details, we define the approximation $Z^{\lambda,\pi,i}$ as a function of both time and state, and associate it with a neural network $ \mathcal{Z}^i$,
\begin{equation*}\label{Network-switch}
Z^{\lambda,\pi,i}(t_k, X^{\pi}_{t_k}) \approx \mathcal{Z}^i(t_k, X^{\pi}_{t_k} \mid \theta^i), \quad k = 0, 1, \dots, N-1,
\end{equation*}
where $\theta^i = \{\theta^i_j\}_{j=1}^{J}$ denotes the vector of trainable parameters within the network $\mathcal{Z}^i$. \par 
Equipping the discretized BSDE (\ref{eq:Y_approx}) with the neural network, we have the numerical scheme given by
\begin{equation*}
\begin{array}{cc}
   \mathcal{V}^{\lambda,\pi,i,\theta}_{t_{k+1}}  =& \mathcal{V}^{\lambda,\pi,i,\theta}_{t_k} - f(t_k, X^{\pi}_{t_k},i)h- \lambda h\sum\limits_{j=1}^m[\mathcal{V}^{\lambda,\pi,j,\theta}_{t_{k}} -k(i,j)-\mathcal{V}^{\lambda,\pi,i,\theta}_{t_k}]^+ \\
&+\mathcal{Z}^{\lambda,\pi,i}(t_k,X^{\pi}_{t_k}\mid \theta^{i})^{\top} ( W_{t_{k+1}} - W_{t_k}), \quad \mathcal{V}^{\lambda,\pi,i,\theta}_0 = v^i.
\end{array}
\end{equation*}

\begin{lemma}\label{th-error bound}
Suppose $b,\sigma$ in $X$ and $f(\cdot,\cdot,i),g(\cdot,i), i\in\Lambda$ in the objective functional follow Assumptions \ref{Assp:1} and \ref{Assp:2} respectively. Let $(\{Y^{\lambda,\pi,i}_{t_k}\}_{i\in \Lambda, k=1,2,...,N},\{Z^{\lambda,\pi,i}_{t_k})\}_{i\in \Lambda, k=1,2,...,N-1})$ and \newline $(\{ \bar{Y}^{\lambda,\pi,i}_{t_k}\}_{i\in \Lambda, k=1,2,...,N},\{\bar{Z}^{\lambda,\pi,i}_{t_k}\}_{i\in \Lambda, k=1,2,...,N-1})$ be two discretizations schemes that follows (\ref{eq:Y_approx})-(\ref{eq:Z_approx}), then
\[
\sup\limits_{0\leq k\leq N} \mathbb{E}[|\Delta \mathbb{Y}^{\lambda,\pi}_{t_k}|] \leq  \mathbb{E}[|\Delta \mathbb{Y}^{\lambda,\pi}_{t_N}|],
\]
where $\Delta \mathbb{Y}^{\lambda,\pi}_{t_k}=(Y^{\lambda,\pi,1}_{t_k}-\bar{Y}^{\lambda,\pi,1}_{t_k},Y^{\lambda,\pi,2}_{t_k}-\bar{Y}^{\lambda,\pi,2}_{t_k},...,Y^{\lambda,\pi,m}_{t_k}-\bar{Y}^{\lambda,\pi,m}_{t_k})^{\top} \in \mathbb{R}^m$.
\end{lemma}
\begin{proof}
$\forall i\in \Lambda$, $k=1,2,...,N$, it follows (\ref{eq:Y_approx}) that, 
\begin{equation*}
\left\{
    \begin{array}{rcl}
   Y^{\lambda,\pi,i}_{t_k}&=& \mathbb{E}_{t_k} \left[ Y^{\lambda,\pi,i}_{t_{k+1}} \right] + h f(t_k, X^{\pi}_{t_k}, i) \nonumber\\
& &+\lambda h \sum\limits_{j=1}^{m} \left[ Y^{\lambda,\pi,j}_{t_k} - k(i,j) - Y^{\lambda,\pi,i}_{t_k}\right]^+,\\
\bar{Y}^{\lambda,\pi,i}_{t_k} &= &\mathbb{E}_{t_k} \left[ \bar{Y}^{\lambda,\pi,i}_{t_{k+1}} \right] + h f(t_k, X^{\pi}_{t_k}, i) \nonumber\\
& &+\lambda h \sum\limits_{j=1}^{m} \left[ \bar{Y}^{\lambda,\pi,j}_{t_k}- k(i,j) - \bar{Y}^{\lambda,\pi,i}_{t_k}\right]^+.
    \end{array}
    \right.
\end{equation*}
It yields that
{\small 
\begin{equation} \label{matrix-del-y}
    \begin{array}{rcl}
  &&\Delta \mathbb{Y}^{\lambda,\pi}_{t_k}\\
  &=& \mathbb{E}_{t_k} \left[ \Delta \mathbb{Y}^{\lambda,\pi} _{t_{k+1}}\right] \\
  &&+
  \left(
  \begin{array}{cc}
  \sum\limits_{j=1}^{m} \lambda h\left[( Y^{\lambda,\pi,j}_{t_k}- k(1,j) - Y^{\lambda,\pi,1}_{t_k})^+-( \bar{Y}^{\lambda,\pi,j}_{t_k}- k(1,j) - \bar{Y}^{\lambda,\pi,1}_{t_k} )^+\right]&\\
 \sum\limits_{j=1}^{m} \lambda h\left[( Y^{\lambda,\pi,j}_{t_k}- k(2,j) - Y^{\lambda,\pi,2}_{t_k})^+-( \bar{Y}^{\lambda,\pi,j}_{t_k}- k(2,j) - \bar{Y}^{\lambda,\pi,2}_{t_k} )^+\right]&  \\
     \cdot  & \\
     \cdot &\\
    \sum\limits_{j=1}^{m} \lambda h\left[( Y^{\lambda,\pi,j}_{t_k}- k(m,j) - Y^{\lambda,\pi,m}_{t_k})^+-( \bar{Y}^{\lambda,\pi,j}_{t_k}- k(m,j) - \bar{Y}^{\lambda,\pi,m}_{t_k} )^+\right]
  \end{array}\right).\\
    \end{array}
\end{equation}
}
The $i$-th line of the matrix is
\begin{equation*}
    \begin{array}{ll}
& \sum\limits_{j=1}^{m}\lambda h\left[ ( Y^{\lambda,\pi,j}_{t_k}- k(i,j) - Y^{\lambda,\pi,i}_{t_k} )^+-( \bar{Y}^{\lambda,\pi,j}_{t_k}- k(i,j) - \bar{Y}^{\lambda,\pi,i}_{t_k} )^+\right]\\
&=\sum\limits_{j=1}^{m} \lambda h\xi^{ij}_{t_k}[( Y^{\lambda,\pi,j}_{t_k}- \bar{Y}^{\lambda,\pi,j}_{t_k})-( Y^{\lambda,\pi,i}_{t_k}- \bar{Y}^{\lambda,\pi,i}_{t_k})],
   \end{array}
\end{equation*}
where 
\begin{equation*}
\xi^{ij}_{t_k}=
\left\{
 \begin{array}{cl}
\frac{( Y^{\lambda,\pi,j}_{t_k}- k(i,j) - Y^{\lambda,\pi,i}_{t_k} )^+-( \bar{Y}^{\lambda,\pi,j}_{t_k}- k(i,j) - \bar{Y}^{\lambda,\pi,i}_{t_k} )^+}{ Y^{\lambda,\pi,j}_{t_k} - Y^{\lambda,\pi,i}_{t_k}-\bar{Y}^{\lambda,\pi,j}_{t_k}+\bar{Y}^{\lambda,\pi,i}_{t_k}}, \ \ &\text{if}\  i\neq j;\\
0,  &\text{if}\  i =j,
   \end{array}
   \right.
\end{equation*}
and $\xi^{ij}_{t_k}\geq 0$ due to the monotomicity of $(\cdot)^{+}$. We will omit the $t_k$ from now on unless it causes confusion.

Finally, we obtain the following vector form of (\ref{matrix-del-y}),
\begin{equation*} 
    \begin{array}{ll}
 \Delta \mathbb{Y}^{\lambda,\pi}_{t_k} =\mathbb{E}_{t_k} \left[ \Delta \mathbb{Y}^{\lambda,\pi} _{t_{k+1}}\right]+\lambda h
  \left(
  \begin{array}{cccccccc}
 -\sum\limits_{j=1}^{m} \xi^{1j},  \xi^{12},..., \xi^{1m}\\
\xi^{21}, -\sum\limits_{j=1}^{m} \xi^{2j},..., \xi^{2m}\\
   \cdot  & \\
     \cdot &\\
  \xi^{m1},\xi^{m2},...,-\sum\limits_{j=1}^{m} \xi^{mj}
  \end{array}\right) \Delta \mathbb{Y}^{\lambda,\pi}_{t_k} .\\
   \end{array}
\end{equation*}
That is,
\begin{equation} \label{vector-y}
    \begin{array}{ll}
\left(I_{m\times m}+\lambda h
  \left(
  \begin{array}{cccccccc}
 \sum\limits_{j=1}^{m} \xi^{1j},  -\xi^{12},..., -\xi^{1m}\\
-\xi^{21}, \sum\limits_{j=1}^{m} \xi^{2j},..., -\xi^{2m}\\
  \cdot  & \\
     \cdot &\\
  -\xi^{m1},-\xi^{m2},...,\sum\limits_{j=1}^{m} \xi^{mj}
  \end{array}\right)  \right) \Delta \mathbb{Y}^{\lambda,\pi}_{t_k} 
  =\mathbb{E}_{t_k} \left[ \Delta \mathbb{Y}^{\lambda,\pi} _{t_{k+1}}\right] .\\
   \end{array}
\end{equation}
Set 
\begin{equation*} 
\begin{array}{ll}
D=\left(
  \begin{array}{cccccccc}
1+\lambda h \sum\limits_{j=1}^{m} \xi^{1j},0,..., 0\\
0, 1+\lambda h\sum\limits_{j=1}^{m} \xi^{2j},..., 0\\
  \cdot  & \\
     \cdot &\\
  0,0,...,1+\lambda h\sum\limits_{j=1}^{m} \xi^{mj}
  \end{array}\right),   \\
   \end{array}
B=\left(
  \begin{array}{cccccccc}
 0,  \lambda h\xi^{12},..., \lambda h\xi^{1m}\\
\lambda h\xi^{21}, 0,..., \lambda h\xi^{2m}\\
  \cdot  & \\
     \cdot &\\
  \lambda h\xi^{m1},\lambda h\xi^{m2},...,0
  \end{array}\right),
\end{equation*}
then (\ref{vector-y}) yields that
\begin{equation*} 
    \begin{array}{rcl}
\Delta \mathbb{Y}^{\lambda,\pi}_{t_k}
& =& \left(D-B  \right) ^{-1}\mathbb{E}_{t_k} \left[ \Delta \mathbb{Y}^{\lambda,\pi} _{t_{k+1}}\right]\\
 &=&D^{-1}(I-BD^{-1})^{-1} \mathbb{E}_{t_k} \left[ \Delta \mathbb{Y}^{\lambda,\pi} _{t_{k+1}}\right].
   \end{array}
\end{equation*}
Taking the $L^1$-norm for the $m$-dimensional vector, we have
\begin{equation} \label{eq-220-3}
    \begin{array}{rcl}
|\Delta \mathbb{Y}^{\lambda,\pi}_{t_k}|
 &\leq & ||D^{-1}(I-BD^{-1})^{-1} || |\mathbb{E}_{t_k} \left[ \Delta \mathbb{Y}^{\lambda,\pi} _{t_{k+1}}\right]|\\
 &\leq &||(I-BD^{-1})^{-1} || \mathbb{E}_{t_k} \left[| \Delta \mathbb{Y}^{\lambda,\pi} _{t_{k+1}}|\right]\\
 &\leq & \frac{1}{1-||BD^{-1}||_{\infty}} \mathbb{E}_{t_k} \left[| \Delta \mathbb{Y}^{\lambda,\pi} _{t_{k+1}}|\right]\\
   \end{array}
\end{equation}
The proof is concluded by taking the expectation on both sides of (\ref{eq-220-3}).
\end{proof}
\begin{theorem}
    Suppose $b,\sigma$ in $X$ and $f(\cdot,\cdot,i),g(\cdot,i), i\in\Lambda$ in the objective functional follow Assumptions \ref{Assp:1} and \ref{Assp:2} respectively. Then
 \begin{align*}
\max_{0\le k\le N}\mathbb{E}\vert \mathcal{V}^{\lambda,\pi,i,\theta}_{t_{k}} - Y^{i}_{t_k}\vert\le C h^{\frac{1}{4}}+  \mathbb{E}\vert \mathcal{V}^{\lambda,\pi,i,\theta}_{t_N}-g(X_{t_N},i)\vert,
 \end{align*}
 where $C$ is a constant depending on $T,X_0.$\par 
\end{theorem}
\begin{proof}
By the triangular inequality we have that
\begin{align}
\mathbb{E}\vert \mathcal{V}^{\lambda,\pi,i,\theta}_{t_{k}} - Y^{i}_{t_k}\vert \le
    \mathbb{E}\vert \mathcal{V}^{\lambda,\pi,i,\theta}_{t_{k}} - Y^{\lambda,\pi,i}_{t_k}\vert +\mathbb{E}\vert Y^{\lambda,\pi,i}_{t_k}- Y^{i}_{t_k}\vert
\end{align}
    Then the result follows from Theorem \ref{thm:discrete_error_revised} and Lemma \ref{th-error bound}.   
\end{proof}

\bibliographystyle{ecta}  

\bibliography{EbertReferences}        
\end{document}